\documentclass[preprint,12pt]{elsarticle}

\usepackage{amssymb,amsmath}\def\d{{\, \rm d}}

\usepackage{amsthm,comment}
\newtheorem{theorem}{Theorem}[section]
\newtheorem{proposition}[theorem]{Proposition}

\begin{document}

\begin{frontmatter}



\title{Launching Drifter Observations in the Presence of Uncertainty}


\author[label1]{Nan Chen\corref{cor1}}
\ead{chennan@math.wisc.edu}
\affiliation[label1]{organization={Department of Mathematics, University of Wisconsin-Madison},
            addressline={480 Lincoln Dr.},
            city={Madison},
            postcode={53706},
            state={WI},
            country={United States of America}}
\author[label2]{Evelyn Lunasin}
\affiliation[label2]{organization={Department of Mathematics, United States Naval Academy},
            addressline={Chauvenet Hall, 572C Holloway Road},
            city={Annapolis},
            postcode={21402-5002},
            state={MD},
            country={United States of America}}
\author[label2,label3]{Stephen Wiggins}
\affiliation[label3]{organization={School of Mathematics, University of Bristol},
            addressline={Fry Building, Woodland Road},
            city={Bristol},
            postcode={BS8 1UG},
            country={United Kingdom}}
\begin{abstract}
Determining the optimal locations for placing extra observational measurements has practical significance. However, the exact underlying flow field is never known in practice. Significant uncertainty appears when the flow field is inferred from a limited number of existing observations via data assimilation or statistical forecast. In this paper, a new computationally efficient strategy for deploying Lagrangian drifters that highlights the central role of uncertainty is developed. A nonlinear trajectory diagnostic approach that underlines the importance of uncertainty is built to construct a phase portrait map. It consists of both the geometric structure of the underlying flow field and the uncertainty in the estimated state from Lagrangian data assimilation. The drifters are deployed at the maxima of this map and are required to be separated enough. Such a strategy allows the drifters to travel the longest distances to collect both the local and global information of the flow field. It also facilitates the reduction of a significant amount of uncertainty. To characterize the uncertainty, the estimated state is given by a probability density function (PDF). An information metric is then introduced to assess the information gain in such a PDF, which is fundamentally different from the traditional path-wise measurements. The information metric also avoids using the unknown truth to quantify the uncertainty reduction, making the method practical. Mathematical analysis exploiting simple illustrative examples is used to validate the strategy. Numerical simulations based on multiscale turbulent flows are then adopted to demonstrate the advantages of this strategy over some other methods.
\end{abstract}


\begin{keyword}
Drifter deployment strategy\sep uncertainty\sep Lagrangian data assimilation\sep Lagrangian descriptor\sep information metric
\MSC 37N10  \sep 93E11  \sep 62E17  \sep 37J25
\end{keyword}

\end{frontmatter}

\section{Introduction}

With the recent advances in various measurement techniques, including satellite images, infrared observatories, drifters, and sensors, observations play an essential role in improving the state estimation and the reconstruction of the flow fields. However, the observational network remains sparse or incomplete in many situations, especially in geophysics and climate science, to accurately recover the entire multiscale flow field. Several data-driven methods have thus been developed to infer the dominant flow structures from data through proper orthogonal decomposition (POD) \cite{willcox2006unsteady, bui2004aerodynamic}, dynamic mode decomposition (DMD) \cite{kramer2017sparse, dang2021dmd} or the discrete empirical interpolation method (DEIM) \cite{drmac2016new, wang2021feasibility}. Sparse optimization is often incorporated into these approaches to avoid overfitting and provide robust results \cite{manohar2018data, herzog2015sequentially, chu2021data}. Recently, neural networks, exploiting deep learning or autoencoder, have also been applied to seek the principle components of the underlying flow from data in a nonlinear way \cite{carlberg2019recovering, fukami2019super, erichson2020shallow}. In addition to these vectorized methods, a tensor-based flow reconstruction from optimally located sensor measurements has been developed that improves the computational efficiency \cite{farazmand2023tensor}. On the other hand, by combining sparse observational data with appropriate physical or statistical models, data assimilation is a widely used tool to infer the underlying flow field via  Bayesian inference \cite{asch2016data, kalnay2003atmospheric, majda2012filtering, law2015data, ghil1991data}. Notably, the observational measurements in data assimilation do not necessarily need to be the velocity field or, in general, the quantity of interest to reconstruct. Different state variables are coupled through the model; therefore, the observational information in one variable can affect the state estimation of the others. This is a unique advantage of data assimilation over many purely data-driven methods that require the flow fields to be directly measured. In addition, the models usually supply crucial prior knowledge of the system, facilitating the accurate recovery of the flow field. Nevertheless, the skill of the inferred flow field relies heavily on the additional information provided by the observations. It is thus significant to determine the most appropriate observational variables and the optimal locations for placing these observational measurements. While the former depends largely on the nature of the physical interests and is often case-dependent, systematic strategies can be developed for optimally choosing the observational sites once the measured quantities are decided.

Depending on how flow fields are measured, observations can be divided into two categories \cite{batchelor1973introduction, lamb1924hydrodynamics}. Eulerian observations are measurements at fixed locations in space through which the fluid flows. In contrast, Lagrangian observations, also known as Lagrangian tracers, are drifters or floaters that collect information when they travel following the movement of the fluids. They become more prevalent with the advance of up-to-date measuring techniques. These Lagrangian observations can supplement traditional Eulerian measurements or be used as stand-alone observations to recover the underlying flow field. Several Lagrangian observational products have been developed in recent years. The Global Drifter Program \cite{centurioni2017global} estimates near-surface currents by tracking the surface drifters deployed throughout the global ocean. The Argo program \cite{gould2004argo} collects information from inside the ocean using a fleet of robotic instruments that drift with the ocean currents and move up and down between the surface and a mid-water level. In addition, sea ice floe trajectories in the marginal ice zone serve as natural Lagrangian data to understand the mesoscale features of the Arctic Ocean \cite{mu2018arctic, chen2022efficient, covington2022bridging}. Other Lagrangian tracers include trash or oil in the ocean \cite{van2012origin, garcia2022structured} and balloons collecting atmospheric data \cite{businger1996balloons}. These Lagrangian observations are widely used for state estimation and prediction in geophysics, climate science, and hydrology \cite{griffa2007lagrangian, blunden2019look, honnorat2009lagrangian, salman2008using, castellari2001prediction}. In particular, the observed Lagrangian trajectories are often combined with models to facilitate the use of data assimilation, known as Lagrangian data assimilation, in recovering the underlying flow field \cite{apte2013impact, apte2008data, apte2008bayesian, ide2002lagrangian}.

The locations for deploying the Lagrangian drifters determine the amount of information of the flow field inferred from these observations, at least within a certain time interval. Strategies of drifter deployment are typically based on the geometric properties of the underlying flow. Launching the drifters in strongly hyperbolic regions was found to provide significant improvement in reconstructing the flow field \cite{poje2002drifter}. A similar conclusion was reached in \cite{salman2008using}, which uses Lagrangian data assimilation to indicate that it is best to target the strongest hyperbolic trajectories for shorter forecasts, although vortex centers can produce good drifter dispersion upon bifurcating on longer timescales. In \cite{treshnikov1986optimal}, a different criterion was developed, which maximizes a collection of functionals based on the spatial degree of the initial data coverage and averaged quantities, including kinetic energy and enstrophy. Machine learning and artificial intelligence have also been used to find suitable strategies. It was shown that the optimal strategy for minimizing a certain norm of velocity error is to deploy the drifters at the center of different segments from a clustering algorithm \cite{tukan2023efficient}. In addition, genetic algorithms were applied, aiming to find the globally optimal drifter launch locations \cite{hernandez1995optimizing, chen2020information}. Other methods in determining the observational locations include exploiting certain low-order reconstructed fields \cite{berliner1999statistical, ballabrera2007observing} or placing the observations aligning with the most sensitive direction \cite{wu2010optimal}.

One of the most significant features and fundamental challenges for deploying Lagrangian drifters in many practical situations is the uncertainty of the underlying flow field. Uncertainty is ubiquitous due to the lack of a full understanding of nature. It arises as a natural outcome since observations are often noisy and the models describing the underlying flow field are rarely perfect. It also appears when ensemble forecast is applied for predicting the future turbulent flow patterns that assist in determining the locations for launching the drifters. Notably, the geometric structure recovered from the inferred flow field when uncertainty is considered can be significantly different from that by ignoring the uncertainty \cite{chen2023lagrangian} and thus leads to a completely different result of deploying drifter observations. However, uncertainty was seldom addressed in most existing work. Many theoretic studies exploit the true flow field, which is unknown in practice, to reveal the geometric structure of the underlying flow that facilitates the analysis of launching drifters \cite{poje2002drifter}. Other studies considered more realistic setups, where the flow field was inferred from Lagrangian data assimilation \cite{salman2008using}. Yet, only the posterior mean estimate was used. Such a path-wise approximation of the truth is valid when the uncertainty in the inference is negligible. However, the resulting uncertainty can be significant in many practical situations when inferring the flow field using only a limited number of existing observations. Therefore, considering the uncertainty in designing strategies for deploying new drifters is essential.

In this paper, a new strategy for deploying Lagrangian drifters that highlights the central role of uncertainty is developed. The underlying flow field is recovered by a certain number of existing drifter observations via a computationally efficient data assimilation algorithm. Fundamentally different from many existing methods, the entire posterior distribution from the Lagrangian data assimilation is considered in this strategy to help determine the locations of launching a set of additional drifters. A recently developed nonlinear trajectory diagnostic approach that underlines the importance of the uncertainty is utilized to construct a phase portrait map \cite{chen2023lagrangian}. It consists of both the geometric structure of the underlying flow field and the uncertainty in the recovered state. The new drifters are launched at the maxima of such a phase portrait map. These locations either correspond to significant flow velocities related to the strong hyperbolic trajectories or contain substantial uncertainties. Adding the new drifters at these locations will allow them to travel long distances, collect global information on the flow field, or advance the reduction of uncertainties. In addition, the areas of the drifter deployment are required to separate enough to avoid the information captured being overlapped.

The new strategy aims to maximize the information gain in the posterior distribution from the Lagrangian data assimilation after the new drifters are launched into the field. As the quantity of interest is the entire distribution for the state estimation, an information metric called relative entropy \cite{kullback1951information, majda2005information, kleeman2011information} is introduced to assess the performance of the strategy. Fundamentally different from the traditional path-wise measurements, the information metric not only quantifies the information captured by the posterior mean state but also evaluates the uncertainty reduction in the covariance or the entire probability density function (PDF). The total information gain is measured by comparing the posterior distribution from the Lagrangian data assimilation with the model equilibrium distribution. The latter is known as the prior distribution, inferred only from the model without exploiting any observational information. Notably, the quantification process does not require the unknown truth, which is more practical and significantly distinguishes it from most path-wise measurements. In addition, the information metric is not directly utilized as the cost function to build a complicated optimization problem. The locations for deploying the new drifters are fully determined by the cheap nonlinear trajectory diagnostic method and the drifter separation criterion, which allow the strategy to be computationally efficient. Finally, although the focus of this work is on maximizing the information gain in the posterior distribution from data assimilation, the framework developed here can be easily applied to reducing the uncertainty in the real-time forecast of the flow field.

The rest of the paper is organized as follows. Section \ref{Sec:Modeling} includes a general stochastic modeling framework for describing random flow fields and the associated Lagrangian data assimilation. The information metric is introduced in Section \ref{Sec:Info}, which is used to quantify the uncertainty reduction when drifters are placed at different locations. Section \ref{Sec:Strategies} presents the mathematical strategy of placing drifter observations. A multiscale turbulent flow is utilized to illustrate the performance of the new strategy, which is shown in Section \ref{Sec:Numerics}. Section \ref{Sec:Conclusion} includes the discussions and conclusion.

\section{Modeling Random Flow Fields and Lagrangian Data Assimilation}\label{Sec:Modeling}
\subsection{Setup of the problem}
Let us begin by stating the setup of the problem. Let us assume there are currently $L_1$ Lagrangian drifters deployed and collecting measurements within the flow field of interest. The trajectories of these $L_1$ drifters are exploited to infer the underlying flow field, which is then used to determine the locations of deploying $L_2$ new drifters at time $t^*$. The goal is to maximize the information gain in the posterior distribution of the data assimilation solution within the time interval $[t^*-\tau,t^*+\tau]$, with $\tau>0$, in light of the $L_1+L_2$ observed Lagrangian trajectories. During this process, once the locations of the $L_2$ new drifters at $t^*$ are determined, the governing equations of these drifters are used to integrate forward and backward in time to create trajectories. These $L_2$ trajectories, together with the existing $L_1$ ones, are then used to compute the posterior distribution within this interval. Although this setup is not real-time, it provides an effective way to study the information gain within $[t^*-\tau,t^*+\tau]$ due to the deployment of the $L_2$ new drifters, where the nonlinear trajectory diagnostic for the phase portrait analysis with uncertainty to determine the locations of launching these drifters is applied to the same interval. The framework can be easily extended to the real-time situation. See the discussions in Section \ref{Sec:Conclusion}. The mathematical model for the flow field and an efficient data assimilation algorithm are described in the following.

\subsection{Modeling random flow fields}
The underlying turbulent flow model utilized in this work is assumed to be given by a finite summation of spectral modes. Random spectral coefficients are adopted to mimic the intrinsic turbulent features in the flow field, which have been widely used in practice \cite{chen2015noisy, majda2003introduction}. In such a modeling framework, the underlying flow velocity field reads
\begin{equation}\label{Ocean_Velocity}
  \mathbf{u}(\mathbf{x},t) = \sum_{\mathbf{k}\in\mathcal{K},\alpha\in\mathcal{A}}\hat{u}_{\mathbf{k},\alpha}(t)e^{i\mathbf{k}\mathbf{x}}\mathbf{r}_{\mathbf{k},\alpha},
\end{equation}
where $\mathbf{x}=(x,y)^\mathtt{T}$ is the two-dimensional coordinate with a double periodic boundary condition $x,y\in[-\pi,\pi]$. There are two indices for the spectral modes. The index $\mathbf{k}=(k_1,k_2)^\mathtt{T}$ is the wavenumber, and the index $\alpha$ represents the characteristic of the mode, including, for example, the gravity modes and the geophysically balanced modes in the study of many geophysical systems. We denote the set of $\alpha$ by $\mathcal{A}$. The set $\mathcal{K}$ consists of the wavenumbers that satisfy $-K_{\mbox{max}}\leq k_1, k_2\leq K_{\mbox{max}}$ with $K_{\mbox{max}}$ being an integer that is pre-determined. The vector $\mathbf{r}_{\mathbf{k},\alpha}$ is the eigenvector, which links the two components of velocity fields $u$ and $v$. For conciseness of notations, the explicit dependence of $\alpha$ on $\hat{u}_{\mathbf{k},\alpha}$ and $\mathbf{r}_{\mathbf{k},\alpha}$ in \eqref{Ocean_Velocity} is omitted in the following discussions. This is a natural simplicity when only one type of characteristic mode is used, which will be the case in the numerical simulations of this work. There, the incompressible flow is considered that includes only the geophysically balanced modes. Therefore, the Fourier coefficient and the eigenvector are written as $\hat{u}_{\mathbf{k}}$ and $\mathbf{r}_{\mathbf{k}}$. Since the left-hand side of \eqref{Ocean_Velocity} is evaluated at physical space, the Fourier coefficient $\hat{u}_{-\mathbf{k}}$ and the eigenvector $\mathbf{r}_{-\mathbf{k}}$ are the complex conjugate of $\hat{u}_{\mathbf{k}}$ and $\mathbf{r}_{\mathbf{k}}$, respectively, for all $\mathbf{k}$. It is worth noting that the framework is not limited to the Fourier basis. Other basis functions and boundary conditions can be utilized in \eqref{Ocean_Velocity} for various applications in practice. Therefore, the representation in \eqref{Ocean_Velocity} is general.

Stochastic models are used to describe the time evolution of each Fourier coefficient $\hat{u}_{\mathbf{k}}$ in \eqref{Ocean_Velocity}, which is a computationally efficient way to mimic the turbulent flow features. Among different stochastic models, the linear stochastic model, namely the complex Ornstein-Uhlenbeck (OU) process \cite{gardiner1985handbook}, is a widely used choice:
\begin{equation}\label{OU_process}
  \frac{\d\hat{u}_{\mathbf{k}}}{\d t} = (- d_\mathbf{k} + i\omega_\mathbf{k}) \hat{u}_{\mathbf{k}} + \mathbf{f}(t) + \sigma_\mathbf{k}\dot{W}_\mathbf{k},
\end{equation}
where $d_\mathbf{k}, \omega_\mathbf{k}$ and $\mathbf{f}(t)$ are damping, phase and deterministic forcing, $\sigma_\mathbf{k}$ is the noise coefficient and $\dot{W}_\mathbf{k}$ is a white noise. The constants $d_\mathbf{k}$, $\omega_\mathbf{k}$ and $\sigma_\mathbf{k}$ are real-valued while the forcings are complex. The stochastic noise in the linear stochastic model is utilized to effectively parameterize the nonlinear deterministic time evolution of chaotic or turbulent dynamics \cite{majda2016introduction, farrell1993stochastic, berner2017stochastic, branicki2018accuracy, majda2018model, li2020predictability, harlim2008filtering, kang2012filtering} such that different Fourier coefficients (excluding those complex conjugate pairs) are independent with respect to each other. This significantly reduces the computational cost as the operations on the summation of complicated nonlinear terms are replaced by a single stochastic term. The decoupled equations for different modes also accelerate the model forecast. These features are particularly useful for efficient data assimilation since the forecast focuses on the statistics instead of the precise value of each single trajectory. Note that the decoupling between different spectral modes does not break the spatial dependence between the state variables at different grid points in physical space, which is automatically recovered after the spatial reconstruction in light of all the spectral modes.

The mathematical framework of modeling random flow fields in \eqref{Ocean_Velocity}--\eqref{OU_process} has been widely applied to studying turbulent flows. Examples include modeling the rotating shallow water equation \cite{chen2015noisy} and the quasi-geostrophic equation \cite{chen2023stochastic}. Such a framework has also been utilized as an effective surrogate forecast model in data assimilation to recover the flow fields associated with the Navier-Stokes equations \cite{branicki2018accuracy}, the moisture-coupled tropical waves \cite{harlim2013test} and a nonlinear topographic barotropic model \cite{chen2023uncertainty}. In addition, the framework has been used to quantify the uncertainty in geophysical turbulent flows \cite{branicki2013non, chen2023uncertainty, chen2016model}.

\subsection{Lagrangian data assimilation}
Lagrangian data assimilation exploits the observed moving trajectories from drifters to infer the underlying velocity field \cite{apte2013impact, apte2008data, apte2008bayesian, ide2002lagrangian}. It is a widely used approach for state estimation and prediction in geophysics, climate science, and hydrology \cite{griffa2007lagrangian, blunden2019look, honnorat2009lagrangian, salman2008using, castellari2001prediction}.

The observational process is given by the evolution equation of the Lagrangian trajectory,
\begin{equation}\label{Tracer_eqn}
\frac{\d\mathbf{x}}{\d t} = \mathbf{u}(\mathbf{x},t) + \sigma_\mathbf{x}\mathbf{W}_\mathbf{x},
\end{equation}
where $\mathbf{W}_\mathbf{x}$ is a two-dimensional real-valued white noise representing the observational uncertainty and small-scale perturbations to the observed drifter trajectories while $\sigma_\mathbf{x}$ is the noise coefficient. The velocity field $\mathbf{u}$ in \eqref{Tracer_eqn} is given by \eqref{Ocean_Velocity}, which is a highly nonlinear function of $\mathbf{x}$. Usually, $L$ equations of \eqref{Tracer_eqn} are used in Lagrangian data assimilation, representing the observed trajectories of $L$ Lagrangian drifters.

Define $\mathbf{X}=(\mathbf{x}_1,\ldots,\mathbf{x}_L)^\mathtt{T}$ the collection of the $L$ observed drifter trajectories and $\mathbf{U}=\{\hat{u}_\mathbf{k}\}$ the vector that collects the Fourier coefficients. In light of \eqref{Ocean_Velocity}, \eqref{OU_process} and \eqref{Tracer_eqn}, the Lagrangian data assimilation can be written in the following form:
\begin{subequations}\label{eq:cgns}
\begin{align}
\frac{\d \mathbf{X}(t)}{\d t} &= \mathbf{A}(\mathbf{X}, t) \mathbf{U}(t) + \sigma_\mathbf{x} \dot{\mathbf{W}}_\mathbf{X}(t),\label{eq:cgns_X}\\
\frac{\d \mathbf{U}(t)}{\d t} &=  \mathbf{F}_\mathbf{U} + \boldsymbol{\Lambda} \mathbf{U}(t)  + \boldsymbol{\Sigma}_\mathbf{U} \dot{\mathbf{W}}_\mathbf{U}(t),\label{eq:cgns_U}
\end{align}
\end{subequations}
where $\mathbf{A}(\mathbf{X}, t)$ contains all the Fourier bases and is, therefore, a highly nonlinear function of $\mathbf{X}$. Despite the strong nonlinearity in the observational process \eqref{eq:cgns_X}, analytic solutions are available for the Lagrangian data assimilation when the linear stochastic models are used as the surrogate forecast model \cite{liptser2013statistics}. This facilitates the state estimation and uncertainty quantification \cite{chen2014information}.

\begin{proposition}[Posterior distribution of Lagrangian data assimilation: Filtering]
Given one realization of the drifter trajectories $\mathbf{X}(s\leq t)$, the filtered posterior distribution $p(\mathbf{U}(t)|\mathbf{X}(s\leq t))$ of Lagrangian data assimilation \eqref{eq:cgns} is conditionally Gaussian,
 where the time evolution of the conditional mean $\boldsymbol\mu$ and the conditional covariance $\bf R$ are given by
\begin{subequations}\label{eq:filter}
\begin{align}
\frac{\d\boldsymbol{\mu}}{\d t} &= \left(\mathbf{F}_\mathbf{U} + \boldsymbol{\Lambda} \boldsymbol{\mu}\right)  + \sigma_\mathbf{x}^{-2}\mathbf{R}\mathbf{A}^\ast\left(\frac{\d \mathbf{X}}{\d t} - \mathbf{A}\boldsymbol{\mu} \right),\label{eq:filter_mu}\\
\frac{\d\mathbf{R}}{\d t} &= \boldsymbol{\Lambda}\mathbf{R} + \mathbf{R}\boldsymbol{\Lambda}^\ast + \boldsymbol{\Sigma}_\mathbf{U}\boldsymbol{\Sigma}_\mathbf{U}^\ast - \sigma_x^{-2}\mathbf{R}\mathbf{A}^\ast\mathbf{A}\mathbf{R},\label{eq:filter_R}
\end{align}
\end{subequations}
with $\cdot^*$ being the complex conjugate transpose.
\end{proposition}
\begin{proof}
The proof can be found in \cite{liptser2013statistics, chen2018conditional}.
\end{proof}

With the filtered solution \eqref{eq:filter} in hand, closed analytic formulae are also available for the smoothing solution.

\begin{proposition}[Posterior distribution of Lagrangian data assimilation: Smoothing]\label{Prop:Smoothing}
Given one realization of the drifter trajectories $\mathbf{X}(t)$ for $t\in[0,T]$, the smoother estimate $p(\mathbf{U}(t)|\mathbf{X}(s), s\in[0,T])\sim\mathcal{N}(\boldsymbol\mu_\mathbf{s}(t),\mathbf{R}_\mathbf{s}(t))$ of the coupled system is also Gaussian,
where the conditional mean $\boldsymbol\mu_\mathbf{s}(t)$ and conditional covariance $\mathbf{R}_\mathbf{s}(t)$ of the smoother satisfy the following backward equations
\begin{subequations}\label{Smoother_Main}
\begin{align}
  \frac{\overleftarrow{\d \boldsymbol{\mu}_\mathbf{s}}}{\d t} &=  -\mathbf{F}_\mathbf{U} - \boldsymbol\Lambda\boldsymbol{\mu}_\mathbf{s}  + (\boldsymbol{\Sigma}_\mathbf{U}\boldsymbol{\Sigma}_\mathbf{U}^*)\mathbf{R}^{-1}(\boldsymbol\mu - \boldsymbol{\mu}_\mathbf{s}),\label{Smoother_Main_mu}\\
  \frac{\overleftarrow{\d \mathbf{R}_\mathbf{s}}}{\d t} &= - (\boldsymbol\Lambda + (\boldsymbol{\Sigma}_\mathbf{U}\boldsymbol{\Sigma}_\mathbf{U}^*) \mathbf{R}^{-1})\mathbf{R}_\mathbf{s} - \mathbf{R}_\mathbf{s}(\boldsymbol\Lambda^* + (\boldsymbol{\Sigma}_\mathbf{U}\boldsymbol{\Sigma}_\mathbf{U}^*)\mathbf{R})  + \boldsymbol{\Sigma}_\mathbf{U}\boldsymbol{\Sigma}_\mathbf{U}^* ,\label{Smoother_Main_R}
\end{align}
\end{subequations}
with $\boldsymbol\mu$ and $\mathbf{R}$ being given by \eqref{eq:filter}. The notation $\overleftarrow{\d \cdot}/\d t$ corresponds to the negative of the usual derivative, which means that the system \eqref{Smoother_Main} is solved backward over $[0,T]$ with the starting value of the nonlinear smoother being the same as the filter estimate $(\boldsymbol\mu_\mathbf{s}(T), \mathbf{R}_\mathbf{s}(T)) = (\boldsymbol\mu(T), \mathbf{R}(T))$.
\end{proposition}
\begin{proof}
The proof can be found in
 \cite{chen2020learning}.
\end{proof}

The smoother estimate \eqref{Smoother_Main} provides a PDF at each time instant for the recovered velocity field, which includes the uncertainty. Given these PDFs and the temporal dependence, an efficient sampling algorithm of the time series of the velocity field $\mathbf{U}$ from the posterior distributions can be developed. The sampled time series of the velocity field will be used to forecast the possible range of the Lagrangian trajectories $\mathbf{x}(t)$ in computing the Lagrangian descriptor in Section \ref{Subsec:LD}.

\begin{proposition}[Sampling trajectories from posterior distributions]\label{Prop:Sampling}
Based on the smoother estimate, an optimal backward sampling of the trajectories associated with the unobserved variable $\mathbf{U}$ satisfies the following explicit formula,
\begin{equation}\label{Sampling_Main}
  \frac{\overleftarrow{\d \mathbf{U}}}{\d t} = \frac{\overleftarrow{\d \boldsymbol\mu_\mathbf{s}}}{\d t} - \big(\boldsymbol\Lambda + (\boldsymbol{\Sigma}_\mathbf{U}\boldsymbol{\Sigma}_\mathbf{U}^*)\mathbf{R}^{-1}\big)(\mathbf{U} - \boldsymbol\mu_\mathbf{s}) + \boldsymbol{\Sigma}_\mathbf{U}\dot{\mathbf{W}}_{\mathbf{U}}(t).
\end{equation}
\end{proposition}
\begin{proof}
The proof can be found in
 \cite{chen2020learning}.
\end{proof}
The temporal dependence in the sampled time series of $\mathbf{U}$ is extremely important. It contains the memory effect of the recovered velocity field, which is a crucial dynamical feature that affects the prediction of the Lagrangian trajectories $\mathbf{X}(t)$. The sampling approach in \eqref{Sampling_Main} fundamentally differs from drawing independent samples at different time instants, giving a noisy time series that lacks the physical properties of $\mathbf{U}$. The sampled trajectories will be used in the Lagrangian descriptor, which provides phase space information to guide optimal placement of additional drifters. The details of Lagrangian descriptor are presented in Section \ref{Subsec:LD}. Note that only the diagonal entries of the matrix $\mathbf{R}$ in \eqref{eq:filter} are saved and applied to computing the smoother and sampling solutions \eqref{Smoother_Main}--\eqref{Sampling_Main}. This significantly reduces computational storage and introduces little error. When the flow field is incompressible, as in the numerical experiment in Section \ref{Sec:Numerics}, it has been shown that $\mathbf{R}$ will converge to a diagonal matrix when $L$ increases \cite{chen2014information}.

In the following, all the numerical experiments for data assimilation will be based on the smoothing solution \eqref{Smoother_Main}. The filtering solution \eqref{eq:filter} will be used in the mathematical analysis of simple examples justifying the proposed strategy of placing drifter observations in Section \ref{Sec:Strategies}. These mathematical results are qualitatively consistent with the numerical outcomes based on the smoothing estimates, but the filtering solutions are much easier to handle mathematically.

\subsection{Spatial reconstruction}
The Lagrangian data assimilation framework \eqref{eq:cgns} focuses on the recovery of each Fourier coefficient. What remains is the reconstruction of the flow field in physical space from the recovered Fourier coefficients. Denote by $\overline{\hat{u}}_\mathbf{k}$ the mean and $\mbox{var}(\hat{u}_\mathbf{k})$ the variance of mode $\mathbf{k}$ from the data assimilation \eqref{Smoother_Main}. Note that the entire posterior covariance $\mathbf{R}_{\mathbf{s}}$ is, in general, not a non-diagonal matrix due to the mixing of the modes in the observation process. Nevertheless, the diagonal components of $\mathbf{R}_{\mathbf{s}}$ usually have more significant amplitudes than the off-diagonal ones, especially when the estimation of $\mathbf{U}$ becomes more accurate \cite{chen2014information}. Therefore, taking the diagonal entries, which represent the actual uncertainty of each mode, to reconstruct the variance at a grid point in physical space is a natural and reasonable choice. The following argument utilizes the mean-fluctuation decomposition of each Gaussian random variable $\hat{u}_\mathbf{k} = \overline{\hat{u}}_\mathbf{k} + \hat{u}_\mathbf{k}^\prime$, where $\overline{\hat{u}}_\mathbf{k}$ is the mean and $\hat{u}_\mathbf{k}^\prime$ is the fluctuation with $\mbox{var}(\hat{u}_\mathbf{k}^\prime) = \mbox{var}(\hat{u}_\mathbf{k})$.

The mean at each grid point is given by
\begin{equation}\label{Ocean_Velocity_mean}
  \overline{u}(\mathbf{x},t) = \sum_{\mathbf{k}\in\mathcal{K}}\overline{\hat{u}}_{\mathbf{k}}(t)e^{i\mathbf{k}\mathbf{x}}\mathbf{r}_{\mathbf{k},1} \qquad \overline{v}(\mathbf{x},t) = \sum_{\mathbf{k}\in\mathcal{K}}\overline{\hat{u}}_{\mathbf{k}}(t)e^{i\mathbf{k}\mathbf{x}}\mathbf{r}_{\mathbf{k},2},
\end{equation}
where $\mathbf{r}_{\mathbf{k},1}$ and $\mathbf{r}_{\mathbf{k},2}$ are the two component of the eigenvector $\mathbf{r}_{\mathbf{k}}$. Similarly, the fluctuation in physical space is given by
\begin{equation}\label{Ocean_Velocity_fluctuation}
  {u}^\prime(\mathbf{x},t) = \sum_{\mathbf{k}\in\mathcal{K}}{\hat{u}}^\prime_{\mathbf{k}}(t)e^{i\mathbf{k}\mathbf{x}}\mathbf{r}_{\mathbf{k},1} \qquad {v}^\prime(\mathbf{x},t) = \sum_{\mathbf{k}\in\mathcal{K}}{\hat{u}}^\prime_{\mathbf{k}}(t)e^{i\mathbf{k}\mathbf{x}}\mathbf{r}_{\mathbf{k},2}.
\end{equation}
Due to the negligible off-diagonal components in the covariance matrix, the variance at a fixed location $\mathbf{x}$ and time $t$ is given by
\begin{equation}\label{Ocean_Velocity_fluctuation}
  \mbox{var}({u}(\mathbf{x},t)) = \sum_{\mathbf{k}\in\mathcal{K}}\mbox{var}(\hat{u}_{\mathbf{k}}(t))(\mathbf{r}_{\mathbf{k},1}\mathbf{r}^*_{\mathbf{k},1}) \qquad \mbox{var}({v}(\mathbf{x},t)) = \sum_{\mathbf{k}\in\mathcal{K}}\mbox{var}(\hat{u}_{\mathbf{k}}(t))(\mathbf{r}_{\mathbf{k},2}\mathbf{r}^*_{\mathbf{k},2}).
\end{equation}
Note that the variances at different grid points are the same. This is a unique feature when global basis functions such as the Fourier bases are used. When ensemble data assimilation methods are applied, the so-called localization technique for recovering the state variables in physical space becomes essential to eliminate the spurious correlations due to the sampling error \cite{petrie2008localization, houtekamer2005ensemble}. In such a case, the spatial distribution of the variance is usually not uniform. Such an inhomogeneous variance distribution is entirely due to numerical approximations. The data assimilation framework developed here avoids such an issue, and the resulting variance represents the exact uncertainty from the Bayesian inference. It is worthwhile to mention that even though the variances at different grid points are the same when the Lagrangian data assimilation is adopted, the change in the Lagrangian descriptor described in Section \ref{Subsec:LD} as a result of such uncertainties is inhomogeneous in space. The change in the Lagrangian descriptor relies on the dynamical properties, and the uncertainty affects the Lagrangian descriptor in a highly nonlinear way.

\section{An Information Metric for Quantifying the Information Gain Using Drifter Observations}\label{Sec:Info}
When uncertainty is considered in recovering the underlying flow field, it is essential to use the entire recovered distribution as the quantity for assessing the state estimation skill. Therefore, information metrics, which directly compare the statistics, are more appropriate choices than the standard path-wise measurements. The latter usually exploits only the posterior mean time series and ignores the information that characterizes the uncertainty.

As the drifter observations provide additional information, the posterior distribution of the inferred flow field via data assimilation is expected to contain less uncertainty than the prior distribution based solely on the model information. For turbulent flows, the prior distribution is typically given by the model equilibrium distribution. The difference between these two distributions is named the uncertainty reduction or the information gain. Denote by $p_t(\mathbf{u})$ the posterior distribution from data assimilation at time $t$ and $p_{eq}(\mathbf{u})$ the prior distribution describing the model equilibrium. The information gain in $p_t(\mathbf{u})$ compared with $p_{eq}(\mathbf{u})$ is given by the relative entropy $\mathcal{P}(p_t(\mathbf{u}),p_{eq}(\mathbf{u}))$  \cite{majda2010quantifying, majda2005information, kleeman2011information},
\begin{equation}\label{Relative_Entropy}
  \mathcal{P}(p_t(\mathbf{u}),p_{eq}(\mathbf{u})) = \int_{\mathbf{u}} p_t(\mathbf{u})\log\left(\frac{p_t(\mathbf{u})}{p_{eq}(\mathbf{u})}\right)\d\mathbf{u},
\end{equation}
which is also known as Kullback-Leibler divergence or information divergence \cite{kullback1951information, kullback1987letter, kullback1959statistics}.
Despite the lack of symmetry, the relative entropy has two attractive features. First, $\mathcal{P}(p_t,p_{eq}) \geq 0$ with equality if and only if $p_t=p_{eq}$. Second, $\mathcal{P}(p_t,p_{eq})$ is invariant under general nonlinear changes of variables. These features provide an attractive framework for assessing the information gain when the drifters are placed differently. A larger value of the relative entropy in \eqref{Relative_Entropy} means the drifter observations play a more significant role in data assimilation. As a remark, the information theory can also be utilized to quantify model error, model sensitivity, and prediction skill \cite{majda2010quantifying, majda2011link, majda2012lessons, branicki2012quantifying, branicki2014quantifying, kleeman2011information, kleeman2002measuring, delsole2004predictability, branicki2013non, branstator2010two}.

One practical setup for utilizing the framework of information theory in many applications arises when both the distributions are Gaussian so that $p_t\sim\mathcal{N}(\bar{\mathbf{u}}_t, \mathbf{R}_t)$ and $p_{eq}\sim\mathcal{N}(\bar{\mathbf{u}}_{eq}, \mathbf{R}_{eq})$, which is the case in the setup of this work. In the Gaussian framework, $\mathcal{P}(p_t(\mathbf{u}),p_{eq}(\mathbf{u}))$ has the following explicit formula \cite{majda2010quantifying, majda2006nonlinear}
\begin{equation}\label{Signal_Dispersion}
\begin{gathered}
  \mathcal{P}(p_t(\mathbf{u}),p_{eq}(\mathbf{u})) = \left[\frac{1}{2}(\bar{\mathbf{u}}_t-\bar{\mathbf{u}}_{eq})^*(\mathbf{R}_{eq})^{-1}(\bar{\mathbf{u}}_t-\bar{\mathbf{u}}_{eq})\right] + \\ \qquad\qquad\qquad \left[-\frac{1}{2}\log\det(\mathbf{R}_t\mathbf{R}_{eq}^{-1}) + \frac{1}{2}(\mbox{tr}(\mathbf{R}_t\mathbf{R}_{eq}^{-1})-\mbox{Dim}(\mathbf{u}))\right],
\end{gathered}
\end{equation}
where $\mbox{Dim}(\mathbf{u})$ is the dimension of $\mathbf{u}$.
The first term in brackets in \eqref{Signal_Dispersion} is called `signal', reflecting the model error in the mean but weighted by the inverse of the model variance, $\mathbf{R}_{eq}$, whereas the second term in brackets, called `dispersion', involves only the model error covariance ratio, $\mathbf{R}_t\mathbf{R}^{-1}_{eq}$. The signal and dispersion terms in (9) are individually invariant under any (linear) change of variables which maps Gaussian distributions to Gaussians.

Figure \ref{fig: Info_gain_schematic} provides a detailed illustration of utilizing the relative entropy to assess the information gain due to the additional knowledge from observations. The true time series is blue in Panels (a)--(c). The red curve in each panel shows the recovered posterior mean time series, while the red shading area shows the two standard deviations associated with the posterior distribution. From Panel (a) to Panel (c), the number of observations increases. Therefore, the posterior mean becomes more accurate, and the posterior variance/standard deviation shrinks. Panels (d)--(e) compare the prior distribution (black), which is the equilibrium distribution of the model, and the posterior distribution (red) at time $t=6.666$. The information gain (IG) is listed at the left top corner, expressed as signal plus dispersion. Both signal and dispersion increase as the distribution moves towards the truth, with the variance decreasing. It is worthwhile highlighting two important facts. First, the relative entropy, by definition, assesses the difference between the posterior and prior distribution rather than directly quantifying the error in the posterior distribution related to the true value. Therefore, a large information gain does not necessarily imply that the posterior distribution is close to the truth. This can be seen by comparing the distributions in green and red in Panel (e), which have the same variance, but the means have different signs. The information gains in these two distributions are the same, but the truth is closer to the center of the distribution in red. Nevertheless, in many applications, especially for Lagrangian data assimilation, when the observational operator is accurate, the posterior mean will converge to the truth as the number of observations increases, and the posterior covariance will shrink to zero \cite{chen2014information}. Therefore, an increase in the information gain often means a nearly simultaneous reduction of path-wise error. Second, the uncertainty reduction reflects not only in the variance but also in the mean. Even if the two distributions have the same variance, if the posterior mean is far from the prior mean, then the two posterior distributions differ from each other. In other words, the posterior distribution contains additional information beyond the prior one. This can be seen by the cyan and red distributions in Panel (b). These two distributions have the same variance, but the cyan one has the same mean as the prior distribution while the red one has a mean value that differs from the prior. Thus, the signal part of the information gain using the posterior distribution in red is nonzero due to the difference in the mean values.

\begin{figure}[htb]\centering
  \hspace*{-2cm}\includegraphics[width=17cm]{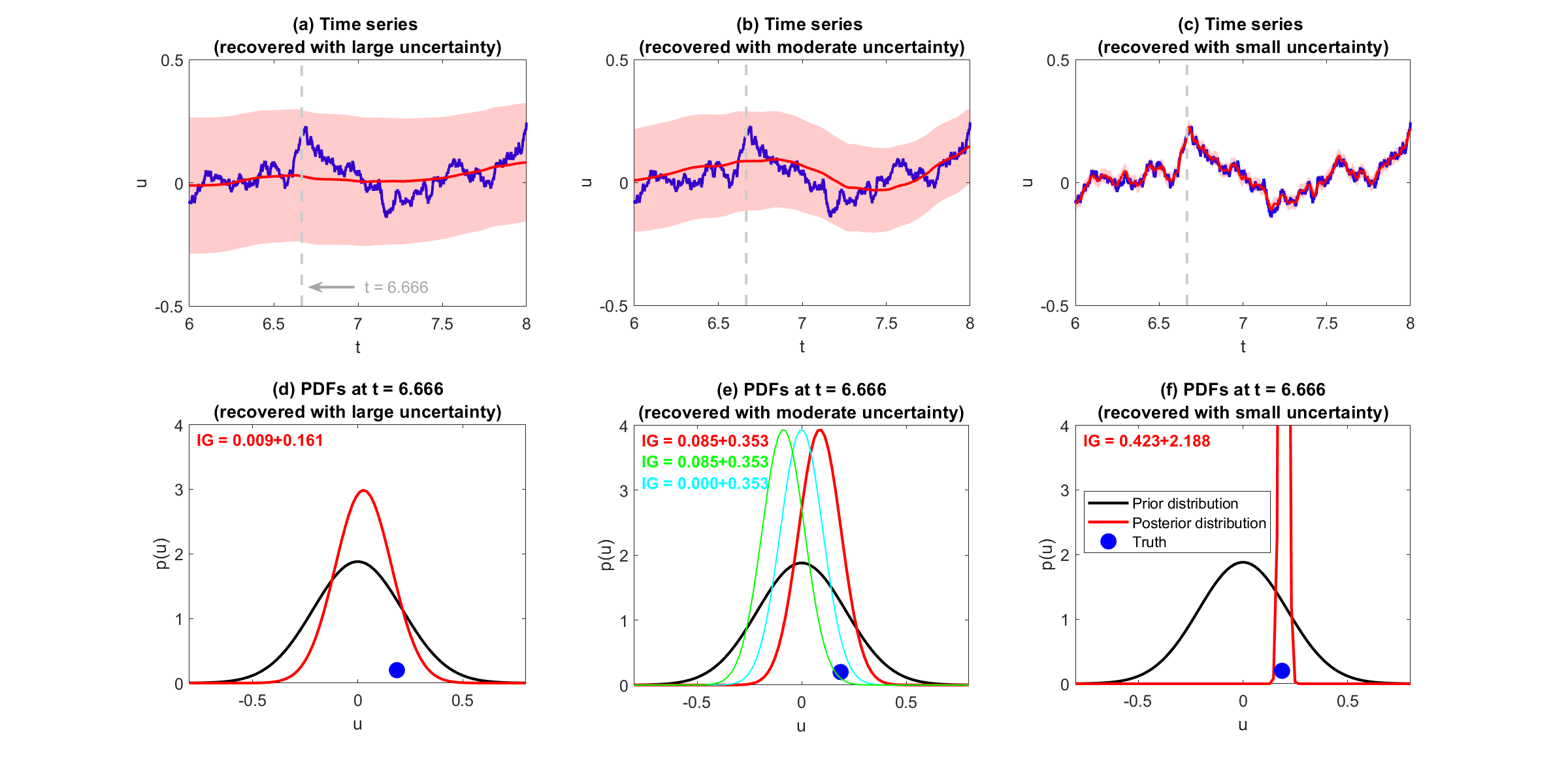}
  \caption{Schematic illustration of using the information measurement, namely the relative entropy \eqref{Relative_Entropy}, for assessing the information gain (IG). Panels (a)--(c): The true time series (blue), the recovered posterior mean time series (red), and the two standard deviations associated with the posterior distribution (red shading). Panel (d)--(f): The prior distribution (black), which is the equilibrium distribution of the model, the posterior distribution (red), and the true value (blue dot) at time $t=6.666$. The information gain (IG) is listed at the left top corner, expressed as signal plus dispersion. The distributions in green and cyan shown in Panel (f) have the same variance as the one in red. The mean of the green one has the negative value as the red while the mean of the cyan is zero.  }\label{fig: Info_gain_schematic}
\end{figure}

For completeness, the root-mean-square error (RMSE) will also be illustrated in some of the following discussions. The RMSE is one of the most commonly used path-wise measurements for time series. It is defined by the averaged error between a certain deterministic estimation and the truth in the following way:
\begin{equation}\label{RMSE_Definition}
    \mbox{RMSE} = \sqrt{\frac{\sum_{i=1}^n(u^{truth}_i-u^{est}_i)^2}{n}}.
  \end{equation}
where $u_i^{truth}$ is the (unknown) true signal evaluated at time $t_i$ with $i=1,\ldots, n$ and $u^{est}_i$ a path-wise estimation. The posterior mean time series is naturally used as $u^{est}_i$ in many applications for evaluating the skill of data assimilation. Note that the application of the RMSE is not limited to time series. It can also be adopted to evaluate the averaged error between two spatiotemporal fields. Such a path-wise measurement is easy to compute and widely used in practice. However, the variance and higher-order moments in the posterior distribution are not used in calculating the RMSE and other path-wise measurements. Therefore, the crucial information that contains the uncertainty is completely missing in these path-wise measurements. This further indicates the necessity of using the information measurement to quantify the uncertainty \eqref{Relative_Entropy}.

\section{The Mathematical Strategy of Placing Drifter Observations}\label{Sec:Strategies}
\subsection{Overview}
The mathematical strategy of placing drifter observations contains two criteria.
First, the drifters are deployed at locations where they can travel long distances within the given time window to collect more information about the flow field. These locations usually have strong velocities that lead to a large signal-to-noise ratio according to the right-hand side of \eqref{Tracer_eqn}. It means the state variables have strong observability, which allows them to be accurately identified. This argument still holds when the velocity is estimated with uncertainty, in which case the drifters are launched at locations with potentially strong velocities from the data assimilation inference. See Section \ref{Subsec:LD}.
Second, it is desirable to place the drifters at locations that are separate from each other. This will prevent the drifters from carrying out similar information if their trajectories nearly overlap. If the drifter trajectories cover the entire domain, then these drifters can collect both local and global information.
Below, simple illustrative examples will be used to support the above arguments. Due to the simple and low-dimensional structure of the flows, mathematical analyses are available to justify both criteria.

\subsection{Mathematical justifications of the two criteria}
This subsection aims to provide mathematical justifications of the two criteria developed above exploiting simple illustrative examples.
\subsubsection{Justification of the first criterion}\label{Subsubsec: criterion1}
Consider a simple flow field with only one mode for both $u$ and $v$,
\begin{equation}\label{flow_one_mode}
   u(x,y,t) = \hat{u}(t)\sin(y) \qquad\mbox{and} \qquad
   v(x,y,t) = \hat{v}(t)\sin(x),
\end{equation}
the associated stream function of which is given by
\begin{equation}
\psi(x,y,t) = \hat{u}(t)\cos(y) - \hat{v}(t) \cos(x).
\end{equation}
The flow field is fully determined once $\hat{u}(t)$ and $\hat{v}(t)$ are given. The true time series of these two variables are driven by two independent real-valued OU processes:
\begin{subequations}\label{flow_one_mode_SDE}
\begin{align}
  \frac{\d\hat{u}}{\d t} &= -d_u \hat{u} + f_u + \sigma_u\dot{W}_u,\label{flow_one_mode_SDE_1}\\
  \frac{\d\hat{v}}{\d t} &= -d_v \hat{v} + f_v + \sigma_v\dot{W}_v,\label{flow_one_mode_SDE_2}
\end{align}
\end{subequations}
and they are given by one random realization from these equations.
Uncertainty appears in data assimilation because of the stochasticity in modeling these two Fourier coefficients.
Now assume only one drifter is utilized, the governing equations of which are
\begin{subequations}\label{flow_one_mode_tracer}
\begin{align}
  \frac{\d x}{\d t} &= \hat{u}(t)\sin(y) + \sigma_x\dot{W}_x,\label{flow_one_mode_tracer_1}\\
  \frac{\d y}{\d t} &= \hat{v}(t)\sin(x) + \sigma_y\dot{W}_y.\label{flow_one_mode_tracer_2}
\end{align}
\end{subequations}
The equations in \eqref{flow_one_mode_tracer} and \eqref{flow_one_mode}--\eqref{flow_one_mode_SDE} correspond to \eqref{Tracer_eqn} (or \eqref{eq:cgns_X}) and \eqref{Ocean_Velocity}--\eqref{OU_process} (or \eqref{eq:cgns_U}), respectively, in the Lagrangian data assimilation framework. For this simple case, the filter solution \eqref{eq:filter} can be written down explicitly. For simplicity, let us consider the equations \eqref{flow_one_mode_SDE_1} and \eqref{flow_one_mode_tracer_1} that describe the motion of $u$. The filter formulae \eqref{eq:filter} for recovering $\hat{u}$ lead to
\begin{subequations}\label{flow_one_mode_filter}
\begin{align}
  \frac{\d \mu}{\d t} &= f_u - d_u \mu + \sigma_x^{-2} r \sin(y) (\sin(y) \hat{u}_{truth} - \sin(y)\mu + \sigma_x\varepsilon_1),\label{flow_one_mode_filter_1}\\
  \frac{\d r}{\d t} &= -2d_ur + \frac{1}{2}\sigma_u^2-\sigma_x^{-2}r^2\sin(y)^2.\label{flow_one_mode_filter_2}
\end{align}
\end{subequations}
where $\mu$ and $r$ are the filter posterior mean and variance of $\hat{u}$, respectively, and $\hat{u}_{truth}$ is the truth of $\hat{u}$. Note that $\d x/\d t$ has been rewritten as $\sin(y) \hat{u}_{truth} + \sigma_x\varepsilon_1$ in \eqref{flow_one_mode_filter_1}, where $\varepsilon_1$ is a standard Gaussian random number at each time instant. Because the last term on the right-hand side of \eqref{flow_one_mode_filter_2} contains $\sin(y)$, which changes in time, the posterior variance $r$ will not converge to a constant, which differs significantly from the standard Kalman (or Kalman-Bucy) filter and is a unique feature of nonlinear data assimilation. Nevertheless, for the illustration here, it is sufficient to study the quantitative behavior of $r$ as a function of $y$. Let $y$ be a constant and consider the solution of
\begin{equation}
0 = -2d_ur + \frac{1}{2}\sigma_u^2-\sigma_x^{-2}r^2\sin(y)^2.
\end{equation}
This is a standard quadratic equation in algebra, which has the solution
\begin{equation}\label{flow_one_mode_r_equilibrium}
r = \frac{\sigma_u^2}{2d_u+\sqrt{4d_u^2+2\sigma_x^{-2}\sin(y)^2\sigma_u^2}},
\end{equation}
where the negative root is discarded since the variance is positive.
As all the other parameters are fixed, $r$ is smaller when $\sin(y)^2$ is larger. Note that a larger $\sin(y)^2$ means a stronger flow velocity along the $x$ direction from \eqref{flow_one_mode}. It also implies that the signal-to-noise ratio in \eqref{flow_one_mode_tracer_1} is more significant, and the system has stronger practical observability. This concludes that the posterior variance will be minimized if the drifter is placed at locations with strong velocities.

Likewise, by re-organizing the right hand side of \eqref{flow_one_mode_filter_1} yields,
\begin{equation}
\frac{\d \mu}{\d t} =  - (d_u + \sigma_x^{-2} r \sin(y)^2  ) \mu + f_u + \sigma_x^{-2} r \sin(y)^2 \hat{u}_{truth}  + \sigma_x^{-2} r \sin(y)\sigma_x\varepsilon_1.
\end{equation}
Again, assume $y$ is a constant. Then the long-term expected value of $\mu$ is
\begin{equation}\label{flow_one_mode_expected_mean}
  \mathbb{E}(\mu) = \frac{ f_u + \sigma_x^{-2} r \sin(y)^2 \hat{u}_{truth}}{d_u + \sigma_x^{-2} r \sin(y)^2}.
\end{equation}
Taking the difference between \eqref{flow_one_mode_expected_mean} and the prior mean $f_u/d_u$ yields
\begin{equation}\label{flow_one_mode_expected_mean_prior}
  \mbox{Diff}_{mean} = \mathbb{E}(\mu)-\frac{f_u}{d_u} = \frac{ f_u + \sigma_x^{-2} r \sin(y)^2 \hat{u}_{truth}}{d_u + \sigma_x^{-2} r \sin(y)^2}-\frac{f_u}{d_u} = \frac{(d_u + f_u \hat{u}_{truth}) \sigma_x^{-2} r\sin(y)^2 }{d_u(d_u + \sigma_x^{-2} r \sin(y)^2)}
\end{equation}
Recall in \eqref{flow_one_mode_r_equilibrium} that $r$ is roughly proportional to the reciprocal of $|\sin(y)|$. Therefore, $\mbox{Diff}_{mean}$ in \eqref{flow_one_mode_expected_mean_prior} can be approximated by
\begin{equation}\label{flow_one_mode_expected_mean_prior2}
  \mbox{Diff}_{mean} \approx \frac{(d_u + f_u \hat{u}_{truth}) \sigma_x^{-2} |\sin(y)| }{d_u(d_u + \sigma_x^{-2} |\sin(y)|)},
\end{equation}
which implies $|\mbox{Diff}_{mean}|$ becomes larger as $|\sin(y)|$. This concludes that the posterior mean differs the most from the prior mean when $|\sin(y)|$ is maximized, corresponding to the case that the flow moves fastest along $x$ direction.

This simple example can also illustrate the criterion affecting the path-wise error. Recall that the governing equation of the truth $\hat{u}_{truth}$ in \eqref{flow_one_mode_SDE_1} is given by
\begin{equation}\label{flow_one_mode_truth}
    \frac{\d\hat{u}_{truth}}{\d t} = -d_u \hat{u}_{truth} + f_u + \sigma_u\varepsilon_2,
\end{equation}
Taking the difference between the governing equations of the posterior mean \eqref{flow_one_mode_filter_1} and the truth $\hat{u}_{truth}$ yields
\begin{equation}\label{flow_one_mode_diff}
    \frac{\d(\hat{u}_{truth}-\mu)}{\d t} = - \big(d_u + \sigma_x^{-2} r \sin(y)^2\big) (\hat{u}_{truth}-\mu)  + \big(\sigma_u\varepsilon_2 - \sigma_x^{-1}r\sin(y) \varepsilon_1\big).
\end{equation}
The equilibrium variance is
\begin{equation}\label{flow_one_mode_equilibrium_var}
  \mbox{var}(\hat{u}_{truth}-\mu) = \frac{\sigma_u^2 + \sigma_x^{-2}r^2\sin(y)^2 }{2\big(d_u + \sigma_x^{-2} r \sin(y)^2\big)}.
\end{equation}
As $r$ is roughly proportional to the reciprocal of $|\sin(y)|$, $\mbox{var}(\hat{u}_{truth}-\mu)$ is scaled as $1/(d_u+c|\sin(y)|)$, where $c$ is a positive constant. As the RMSE in $\mu$ compared with $\hat{u}_{truth}$ is proportional to $\mbox{var}(\hat{u}_{truth}-\mu)$, it is concluded that when the drifter stays at locations with strong velocity the resulting RMSE remains small.

As a numerical validation, consider a set of numerical simulations with $d_u=d_v=0.5$, $f_u=f_v=0$, $\sigma_u=\sigma_v=0.3$, $\sigma_x=\sigma_y=0.003$. Take different values of $y$ at $t=5$ and then generate its trajectory within the interval $t\in[4,6]$ by running the drifter equations forward and backward. Within such a short interval, the drifter trajectory will stay around the locations at $t=5$. These trajectories will be used for Lagrangian data assimilation. Note that the numerical simulation here shows the smoothing state estimation while the analysis above utilized the simpler case, namely the filter estimation, for the sake of deriving the analytic expressions. The other component of the drifter location at $t=5$ is set to be $x=0$, which has little impact on recovering $\hat{u}$. Panels (a)--(d) in Figure \ref{fig: One_mode_example} show the negative RMSE, the information gain, and its signal and dispersion parts. The negative sign is imposed in front of the RMSE such that the color maps of different metrics have similar representations --- larger values mean more skillful. The RMSE error and the information gain reach the optimal points when $|\sin(y)|=1$. Panels (e)--(f) show the posterior mean and posterior variance when the $y$ coordinate of the drifter stays around $y=\pi/2$ ($\sin(y)=1$) and $y=0$ ($\sin(y)=0$), where their trajectories within the interval $t\in[4,6]$ are shown in Panel (f). Deploying the drifter at $y=\pi/2$ leads to significantly much better results in terms of both the reduction of the uncertainty and minimizing the path-wise error.

To summarize, when the drifter is placed and stays at the locations with strong velocity fields, the information gain in both the signal and dispersion parts will be maximized. Meanwhile, the path-wise error in the posterior mean will also be minimized.

\begin{figure}[htb]\centering
  \hspace*{-2cm}\includegraphics[width=17cm]{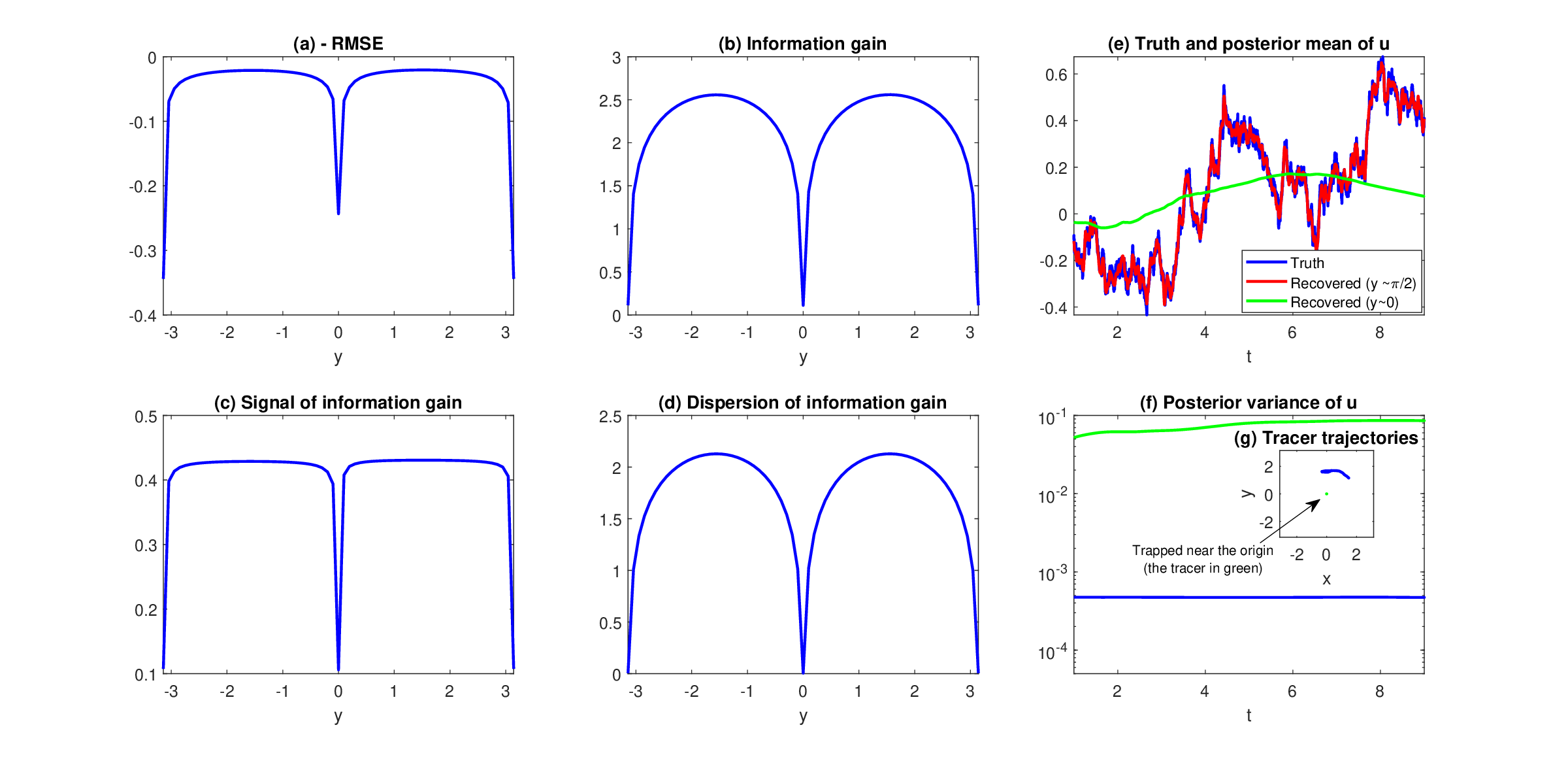}
  \caption{Numerical results of the flow model with one mode in each direction described in Section \ref{Subsubsec: criterion1}. Panel (a): The negative of the RMSE between the truth and the posterior mean of $\hat{u}$. Panel (b): The total information gain via the relative entropy \eqref{Relative_Entropy}. Panel (c): The signal part of the total information gain in the posterior distribution of $\hat{u}$ related to the prior distribution. Panel (d): The dispersion part of the relative entropy. Panel (e): Comparison between the truth and the posterior mean time series. Panel (f): The posterior variance. Note the logarithm scale in the y-axis. Panel (g): the two drifter trajectories within $t\in[4,6]$.   }\label{fig: One_mode_example}
\end{figure}

\subsubsection{Justification of the second criterion}\label{Subsubsec: criterion2}
To justify the second criterion, consider the flow field with two modes along each direction. The stream function $\psi$ and the velocity $(u, v)$ for the underlying flow field are given by
\begin{equation}\label{flow_two_modes}
\begin{aligned}
  \psi(x,y,t) &= \frac{1}{2}\left(\hat{u}(t)e^{iy} + \hat{v}(t)e^{ix} + c.c.\right),\\
   u(x,y,t) &= -\frac{1}{2}i\hat{u}(t)e^{iy} + c.c. = \hat{u}^R(t) \sin(y) + \hat{u}^I(t)\cos(y),\qquad\mbox{and}\\
   v(x,y,t) &= \frac{1}{2}i\hat{v}(t)e^{ix} + c.c. = - \hat{v}^R(t) \sin(x) - \hat{v}^I(t)\cos(x),
\end{aligned}
\end{equation}
where c.c. means the complex conjugate and $\hat{u}=\hat{u}^R+i\hat{u}^I$ and $\hat{v}=\hat{v}^R+i\hat{v}^I$. The two Fourier coefficients $\hat{u}(t)$ and $\hat{v}(t)$ are driven by complex OU processes,
\begin{subequations}\label{flow_two_modes_SDE}
\begin{align}
  \frac{\d\hat{u}}{\d t} &= -d_u \hat{u} + f_u + \sigma_u\dot{W}_u,\label{flow_two_modes_SDE_1}\\
  \frac{\d\hat{v}}{\d t} &= -d_v \hat{v} + f_v + \sigma_v\dot{W}_v,\label{flow_two_modes_SDE_2}
\end{align}
\end{subequations}
where $f_u, f_v, \dot{W}_u$ and $\dot{W}_v$ are complex-valued.
As in Section \ref{Subsubsec: criterion1}, let us focus on the velocity component $u$. To illustrate the importance of separate drifter locations, consider the deployments of two drifters along $x$ direction, which are denoted by $x_1$ and $x_2$,
\begin{equation}\label{flow_two_modes_tracers}
\begin{aligned}
\frac{\d x_1}{\d t} &= \hat{u}^R(t) \sin(y_1) + \hat{u}^I(t)\cos(y_1) + \sigma_x\dot{W}_{x,1},\\
\frac{\d x_2}{\d t} &= \hat{u}^R(t) \sin(y_2) + \hat{u}^I(t)\cos(y_2) + \sigma_x\dot{W}_{x,2}.
\end{aligned}
\end{equation}
Consider the coefficient matrix from the right-hand side of \eqref{flow_two_modes_tracers},
\begin{equation}\label{flow_two_modes_matrix}
  M = \left(
        \begin{array}{cc}
          \sin(y_1) & \cos(y_1) \\
          \sin(y_2) & \cos(y_2) \\
        \end{array}
      \right).
\end{equation}
It is anticipated that the eigenvalues of $M$ should have large values such that the signal-to-noise ratio on the right-hand side of \eqref{flow_two_modes_tracers} becomes significant. Denote by $\lambda_1$ and $\lambda_2$ the two eigenvalues of $M$. Panel (a) and Panel (b) of Figure \ref{fig: Two_modes_example} show $|\lambda_1| + |\lambda_2|$ and $|\lambda_1\cdot\lambda_2|$, respectively, as a function of $y_1$ and $y_2$. A similar profile is obtained when computing the hyperellipse expanded by the two eigenvalues, which is omitted here. All these results indicate that if $y_1$ and $y_2$ are very close, then the summation, the product, and the ellipse area given by the absolute value of the two eigenvalues will be small. Therefore, these are not suitable locations for placing the drifters, as noise will dominate the useful information resulting in weak observability of the system.

To validate such a mathematical justification, consider a set of numerical simulations with $d_u=d_v=0.5$, $f_u=f_v=0$, $\sigma_u=\sigma_v=0.3$, $\sigma_x=\sigma_y=0.003$. Take different values of $y_1$ and $y_2$ at $t=5$ and then generate their trajectories within the interval $t\in[4,6]$ by running the drifter equations forward and backward. These trajectories will be used for Lagrangian data assimilation (smoothing state estimation). The other component of the drifter locations at $t=5$ is set to be $x_1=\pi$ and $x_2=\pi/2$, which again has little impact on recovering $\hat{u}$. Panels (c) and (d) show the negative RMSE and the relative entropy as a function of $y_1$ and $y_2$ at $t=5$. Panels (e)--(f) show the information gain in the signal and dispersion parts. The results in Panels (c)--(e) are consistent with those in Panels (a) and (b). They suggest that the two drifters should be separated enough to avoid duplicating the information. In practice, the number of drifters is often smaller than the degree of freedom of the underlying flow. Therefore, drifters are required to be separated enough to capture as large a spanned subspace as possible. The results here also have another important implication. There are broad areas in the phase space of $(y_1,y_2)$, where the information gain is nearly maximum. This means there is no need to design a delicate strategy for drifter placement. Instead, by guaranteeing that the drifters are distanced from each other, random deployment of drifters at locations with strong underlying flow fields will be a practically simple and useful strategy.

\begin{figure}[htb]\centering
  \hspace*{-2cm}\includegraphics[width=17cm]{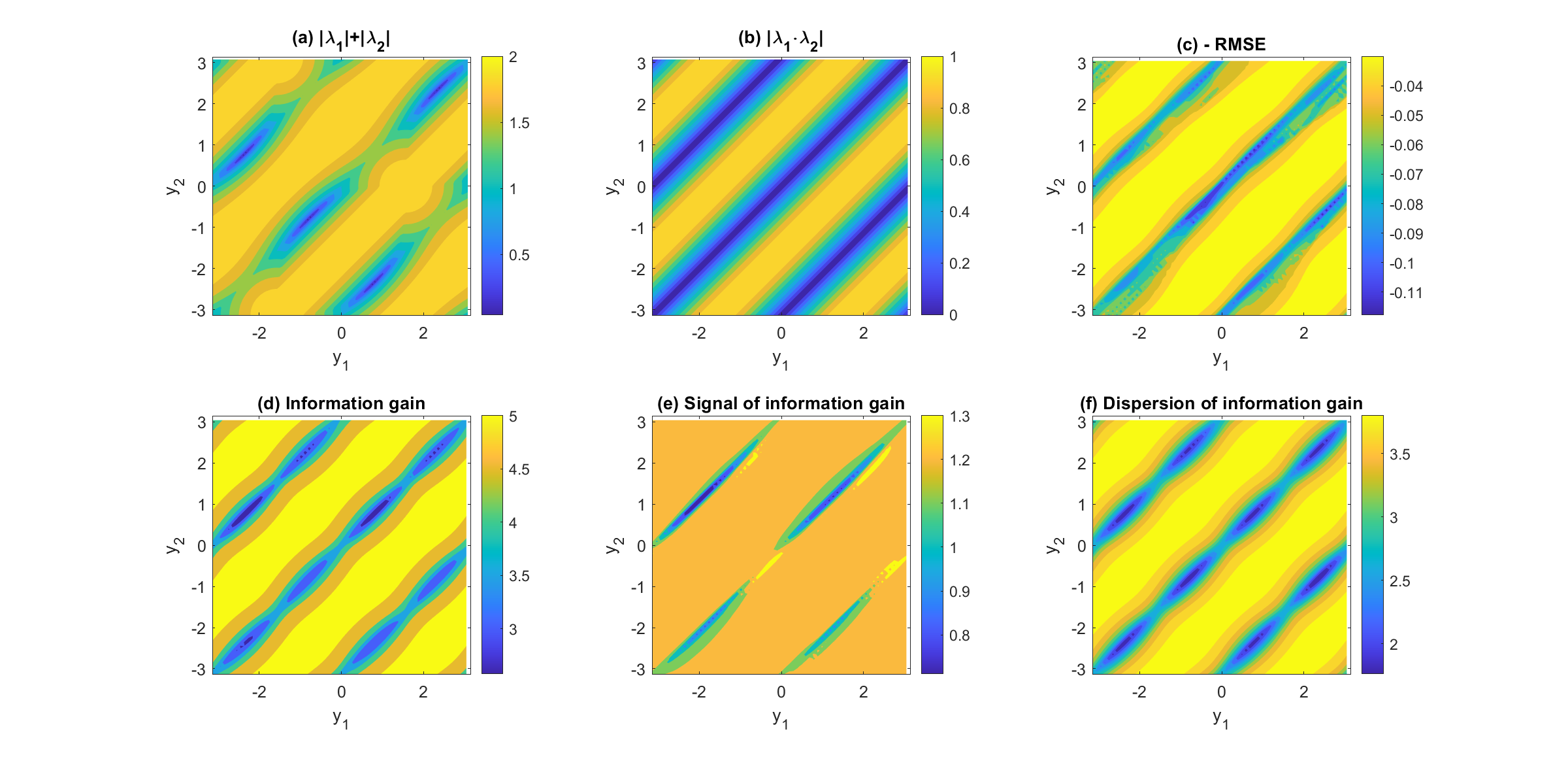}
  \caption{Numerical results of the flow model with two modes in each direction described in Section \ref{Subsubsec: criterion2}. Panels (a)--(b): $|\lambda_1|+|\lambda_2|$ and $|\lambda_1\cdot\lambda_2|$, where $\lambda_1$ and $\lambda_2$ are the two eigenvalues of $M$ in \eqref{flow_two_modes_matrix}. Panel (c): The negative of RMSE between the truth and $\hat{u}$. Panels (d)--(f): The information gain in the posterior distribution of $\hat{u}$ related to the prior distribution and the associated signal and dispersion parts. }\label{fig: Two_modes_example}
\end{figure}

\subsection{Using Lagrangian descriptor to determine the drifter deployment}\label{Subsec:LD}
Recall the first criterion: drifters traveling long distances usually carry more information. As the flow field is time-dependent and turbulent, the location with the maximum velocity at a given time instant $t^*$ does not guarantee that putting the drifter there will result in its traveling a long distance within the time window $[t^*-\tau,t^*+\tau]$. To find the locations at $t^*$ where the deployed drifters travel the longest distances within the time window $[t^*-\tau,t^*+\tau]$, a nonlinear trajectory diagnostic technique --- the Lagrangian descriptor --- is adopted.

Denote by $\mathbf{x}=(x,y)^\mathtt{T}$ the two-dimensional displacement and $\mathbf{u}=(u,v)^\mathtt{T}$ the two-dimensional velocity field.
The general formula of the Lagrangian descriptor is as follows \cite{mancho2013lagrangian, lopesino2017theoretical, garcia2022lagrangian}
\begin{equation}\label{LD_General_Formula}
  \mathcal{L}(\mathbf{x}^*,t^*) = \int_{t^*-\tau}^{t^*+\tau} F(\mathbf{x}, t) \d t,
\end{equation}
where $F=|\tilde{F}|$ is a scalar field with positive values and $t$ is time. According to \eqref{LD_General_Formula}, $\mathcal{L}$ is the integrated modulus of $\tilde{F}$
along a trajectory from the past $t^*-\tau$ to the future $t^*+\tau$ that goes through a point $\mathbf{x}^*$ at time $t^*$. In many applications, the variable $t^*$ is fixed, and therefore the Lagrangian descriptor gives a two-dimensional contour map on the mesh grids of $x^*$ and $y^*$. One commonly used Lagrangian descriptor is by taking $F$ to be the arc length of the path traced by the trajectory, which is also the Lagrangian descriptor used in this work:
\begin{equation}\label{LD_VelocityBased}
M_{vel}(\mathbf{x}^*,t^*) = \int_{t^*-\tau}^{t^*+\tau} \sqrt{\left(\frac{\partial x}{\partial t}\right)^2+\left(\frac{\partial y}{\partial t}\right)^2} \d t = \int_{t^*-\tau}^{t^*+\tau} \sqrt{u^2+v^2} \d t.
\end{equation}
Once the Lagrangian descriptor is computed, it is usually normalized to its maximum value in space for illustration purposes. The expression in \eqref{LD_VelocityBased} explicitly depends on the velocity field while the integration is along the Lagrangian trajectory. In the standard definition in \eqref{LD_General_Formula}, the velocity field $\mathbf{u}$ and the trajectory $\mathbf{x}$ are both deterministic. However, the recovered solution from Lagrangian data assimilation contains uncertainty. Consequently, the displacement $\mathbf{x}$, driven by $\mathbf{u}$, also becomes non-deterministic.

In the presence of uncertainty, the Lagrangian descriptor in \eqref{LD_VelocityBased} is revised by taking the expectation in terms of both $\mathbf{u}$ and $\mathbf{x}$. The expectation of $\mathbf{x}$ accounts for the uncertainty of where the Lagrangian trajectories are located. The expectation of $\mathbf{u}$ is for evaluating the integrand $\sqrt{u(\mathbf{x},t)^2+v(\mathbf{x},t)^2}$ in computing Lagrangian descriptor at each possible fixed location. The former is natural, which has been considered in the previous works for analyzing Lagrangian coherent structures using different Eulerian and Lagrangian methods \cite{badza2023sensitive, rapp2020uncertain}. The latter is unique for the Lagrangian descriptor. Taking into account these uncertainties, the Lagrangian descriptor is given by
\begin{equation}\label{LD_VelocityBased_UQ}
\begin{aligned}
M_{vel}^{UQ}(\mathbf{x}^*,t^*) &= \mathbb{E}_{\mathbf{u},\mathbf{x}}\left[\int_{t^*-\tau}^{t^*+\tau} \sqrt{u(\mathbf{x},t)^2+v(\mathbf{x},t)^2} \d t\right]\\
&= \int_{t^*-\tau}^{t^*+\tau} \int_{\mathbf{x}}\int_{\mathbf{u}}\sqrt{u(\mathbf{x},t)^2+v(\mathbf{x},t)^2}p(\mathbf{u}|\mathbf{x})p(\mathbf{x}) \d\mathbf{u}\d\mathbf{x}\d t.
\end{aligned}
\end{equation}
Note that $\mathbf{u}$ is a function of $\mathbf{x}$ as the velocity depends on the location. Once $t$ is given, the distribution of $\mathbf{u}$ is obtained from the Lagrangian data assimilation. On the other hand, computing the probability of the forward or backward path $\mathbf{x}$ depends on the initial condition $\mathbf{x}^*$, time $t^*$ and the underlying velocity field $\mathbf{u}$. Thus, \eqref{LD_VelocityBased_UQ} is rewritten as
\begin{equation}\label{LD_VelocityBased_UQ_2}
M_{vel}^{UQ}(\mathbf{x}^*,t^*) = \int_{t^*-\tau}^{t^*+\tau} \left(\int_{\mathbf{x}}\mathbb{E}_{\mathbf{u}}\left[\sqrt{u(\mathbf{x},t)^2+v(\mathbf{x},t)^2}\right]p(\mathbf{x}) \d\mathbf{x}\right)\d t.
\end{equation}

In a recent work \cite{chen2023lagrangian}, an analytic approximation has been developed to compute
\begin{equation}\label{LD_VelocityBased_UQ_2_partI}
  \mathcal{F}_E(\mathbf{x},t) = \mathbb{E}_{\mathbf{u}}\left[\sqrt{u(\mathbf{x},t)^2+v(\mathbf{x},t)^2}\right]
\end{equation}
in \eqref{LD_VelocityBased_UQ_2}. The derivation of the analytic expression assumes that both $u$ and $v$ are Gaussian distributed. This is precisely the case in the Lagrangian data assimilation framework here and are good approximations when ensemble data assimilation is adopted.
With the $F_E(\mathbf{x},t)$ being computed, what remains is to estimate $p(\mathbf{x})$ to finish calculating the Lagrangian descriptor \eqref{LD_VelocityBased_UQ_2}. Unlike the velocity field, the distribution of the trajectory at a given time instant $p(\mathbf{x})$ is generally non-Gaussian. This can be seen by noting that the domain has finite support while the support of a Gaussian distribution is infinite. In addition, the governing equation of the drifter is not linear, which does not guarantee that the resulting distribution is Gaussian. Therefore, unlike the usual way of handling  underlying the velocity field, applying a direct mean-fluctuation decomposition assuming a Gaussian distribution for $p(\mathbf{x})$ is inappropriate. In general, $p(\mathbf{x})$ can be estimated by first sampling $\mathbf{u}$ and then plugging the resulting $\mathbf{u}$ to the drifter equation that gives a set of $\mathbf{x}$. A standard two-dimensional kernel density estimation with Gaussian kernels is adopted to provide an analytic expression of $p(\mathbf{x})$. The bandwidth is given by the rule-of-thumb bandwidth estimator for the two dimensions independently \cite{silverman1986density}.

In the context of the Lagrangian data assimilation in Section \ref{Subsec:LD}, the sampled time series of $\mathbf{u}$ can be obtained using the analytic formula in Proposition \ref{Prop:Sampling}. Since $\mathbf{u}$ is written in spectral form, sampling or forecasting the coefficients for different modes can be carried out independently.

Figure \ref{fig: LD_Illustration} explains how the uncertainty affects the behavior of the Lagrangian descriptor. In this numerical example, the true flow field is the same as the one used in Section \ref{Subsubsec: criterion2} but with different parameters. The three columns of the figure display three scenarios. The red dashed line in Panels (a)--(c) shows the unknown true $\hat{u}(t)$, which is either $\hat{u}(t)\equiv0$ (Scenarios I and III) or $\hat{u}(t)\equiv2.5$ (Scenario II). In this illustration example, the uncertainty is not computed precisely from the Lagrangian data assimilation. Instead, the uncertainty is prescribed as a constant over time for simplicity. The uncertainty at each time instant satisfies a Gaussian distribution. The mean equals the truth, and the standard deviation is prescribed to be $0.16$ (Scenarios I and II) or $1.6$ (Scenario III). Given the uncertainty, $50$ time series of $\hat{u}(t)$ are sampled, and $10$ of them are shown in black curves in the three panels. A similar manipulation is done for $\hat{v}(t)$. Starting from $(x(0),y(0))=(0,0)$, Panels (d)--(f) show the $50$ Lagrangian trajectories $(x(t),y(t))$ up to $t=1$. Lagrangian descriptor leads to small and large values for Scenarios I and II, which are nearly deterministic. In Scenario III, although the truth is zero, multiple sampled time series of $\hat{u}(t)$ have high values due to uncertainty. In other words, the inference indicates that the recovered velocity fields $\hat{u}(t)$ are potentially very strong. Thus, the Lagrangian trajectories corresponding to these sampled velocity fields can also travel long distances, leading to a large value of the Lagrangian descriptor \eqref{LD_VelocityBased_UQ}. Clearly, new drifters should not be placed where the Lagrangian descriptor gives a small value, such as in Scenario I. In contrast, new drifters are suggested to place at locations with large values computed from the Lagrangian descriptor. Note that the Lagrangian descriptor gives large values in both Scenario II and III but has different mechanisms. On the one hand, the large value computed from the Lagrangian descriptor in Scenario II is due to the strong true velocity field. These are the locations where the drifters are deployed as per the design principle. On the other hand, the large Lagrangian descriptor value in Scenario III is due to the large uncertainty that leads to potentially strong velocity fields. This remains consistent with the criterion that the drifters are deployed at locations with `strong velocity fields’. Yet, the velocity here is not the true velocity but the inferred velocity in the presence of uncertainty. For this reason, deploying drifters there is expected to reduce uncertainty in the surrounding area. Since the truth is unknown in practice, the two scenarios are indistinguishable from values resulting from the Lagrangian descriptor. Regardless of the two situations, deploying drifters at such locations facilitates the recovery of the flow field. In a typical case with global basis functions as the setup in this work, the Lagrangian descriptor rarely results in a significantly larger value at the locations with a static flow than that with a moderate or strong flow velocity. Drifters at locations suggested by the Lagrangian descriptor usually travel relatively long distances and reduce uncertainty. It is worth remarking that if localization is incorporated into the ensemble data assimilation, then the uncertainty at the locations with no existing drifters around can be as large as the uncertainty in the prior distribution. In such a case, those locations will be the maxima from the Lagrangian descriptor. It is natural to place drifters there to collect the missing local information.

\begin{figure}[htb]\centering
  \hspace*{-2cm}\includegraphics[width=17cm]{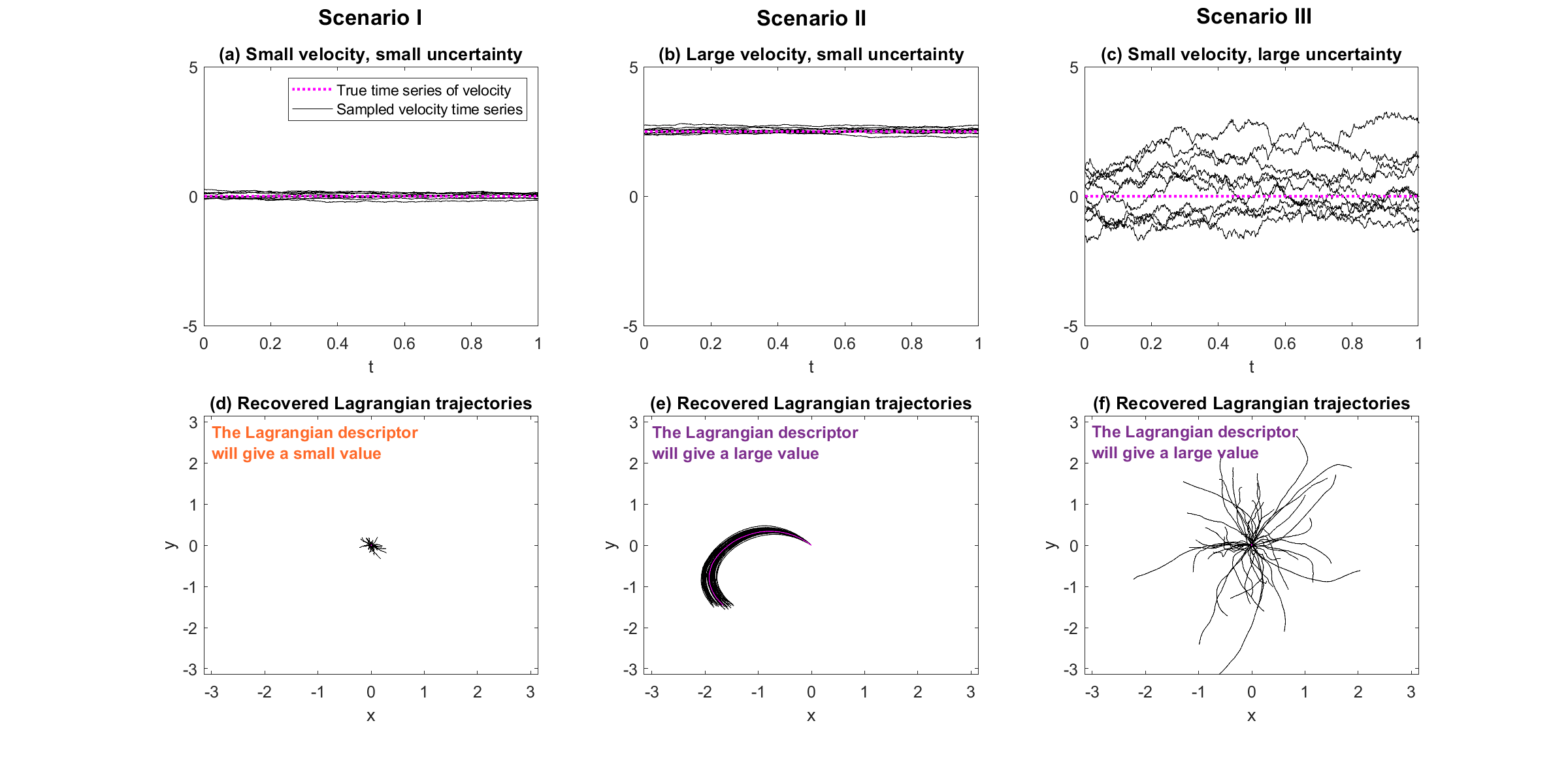}
  \caption{Illustration of how the uncertainty affects the behavior of the Lagrangian descriptor. The true flow field is the same as the one used in Section \ref{Subsubsec: criterion2} but with different parameters. Panels (a)--(c): The unknown true $\hat{u}(t)$ (red) and $10$ out of $50$ sampled velocity time series (black) due to the prescribed uncertainty of $\hat{u}(t)$. Panels (d)--(f): The $50$ Lagrangian trajectories $(x(t),y(t))$ up to $t=1$. }\label{fig: LD_Illustration}
\end{figure}

Finally, regarding the second criterion of separating drifters, a minimum distance $D_{min}$ is prescribed between the new drifters and the distance between each new drifter and all the existing ones. This minimum distance is chosen empirically which can be determined by the minimum resolved spatial scale. Then the new drifters are deployed at the maxima resulting from the Lagrangian descriptor.

\section{Application to Multiscale Turbulent Flows}\label{Sec:Numerics}
\subsection{Setup}
The underlying flow model is given by \eqref{Ocean_Velocity}--\eqref{OU_process} with doubly periodic boundary conditions. The flow field is assumed incompressible, and no mean background flow is included. The maximum Fourier wave number for both $k_1$ and $k_2$ is taken to be $K_{\mbox{max}}=4$ such that there are a total of $80$ Fourier modes. The parameters are
\begin{equation}\label{Parameters_eddy_model}
  d_\mathbf{k} = 0.5,\qquad \omega_\mathbf{k}=0,\qquad f_\mathbf{k}=0\qquad \mbox{and}\qquad \sigma_\mathbf{k} = 0.5.
\end{equation}
for all $\mathbf{k}$ such that the flow field has an equipartition of the energy. The initial distribution of drifters is uniform, consistent with the statistical equilibrium state \cite{chen2014information}. The time instant $t^*=5$ is chosen for deploying new drifters. Note that since the decorrelation time of all Fourier coefficients is only $2$ time units, the initial value has little impact on the flow field at $t^*=5$. Two sets of experiments are studied here. In the first case, there are $L_1=32$ existing drifters, and the goal is to place $L_2 = 6$ additional drifters aiming to maximize the information gain within the time interval $t\in[4,6]$. With a relatively large number of drifters, the posterior mean is expected to be pretty accurate, and the uncertainty reduction mainly lies in the covariance. In the second case, the number of the existing drifters is reduced to $L_1=12$ such that biases appear in the mean and large uncertainty emerges in the covariance. The goal is to maximize the averaged information gain within the time interval $t\in[4.5,5.5]$. Note that the size of the interval is shortened due to the larger uncertainty. This mimics the realistic situations where larger initial uncertainty means shortening potential predictability. When the Lagrangian descriptor is applied to determine the locations for deploying the $L_2$ new drifters, the same time interval as the target for the information gain is chosen for the associated path integration. Once the new drifters are placed at $t^*=5$, the governing equations of the drifters are integrated forward and background to create the corresponding Lagrangian trajectories, which are then used as the observations for the data assimilation.

The procedure for the experiments are summarized as follows.
\begin{enumerate}
  \item [Step 1.] Use the flow model \eqref{Ocean_Velocity}--\eqref{OU_process} to generate the true underlying flow field for $t\in[0,T]$.
  \item [Step 2.] Use the drifter equation \eqref{Tracer_eqn} to create $L_1$ Lagrangian trajectories within the same interval.
  \item [Step 3.] Apply the Lagrangian data assimilation \eqref{Smoother_Main}--\eqref{Sampling_Main} to recover the underlying flow field, including the associated uncertainty.
  \item [Step 4.] Apply the Lagrangian descriptor \eqref{LD_VelocityBased_UQ} to compute the two-dimensional map, where the path integration is taken within the interval $[t^*-\tau,t^*+\tau]$ that is a subset of $[0,T]$.
  \item [Step 5.] Based on the solution from the Lagrangian descriptor and the distance criterion with a prescribed $D_{min}$, deploy $L_2$ additional drifters at time $t^*$.
  \item [Step 6.] Starting from $t^*$, integrate backward and forward using \eqref{Tracer_eqn} to create the trajectories of these $L_2$ drifters.
  \item [Step 7.] Apply Lagrangian data assimilation \eqref{Smoother_Main} using all these $L_1+L_2$ drifters to obtain the posterior distribution.
  \item [Step 8.] Use the information metric \eqref{Signal_Dispersion} to compute the information gain.
\end{enumerate}
The initial condition of data assimilation for all modes is chosen to be an independent standard Gaussian distribution $\mathcal{N}(0,0.1^2)$. When $t^*-\tau$ is far from the initial time $t=0$, the data assimilation results have almost no dependence on the initial condition and are fully determined by balancing the model and observational uncertainties. A similar reason applies to the other side of the interval.  When $t^*+\tau$ is sufficiently distanced from the endpoint $t=T$, the smoother estimate at $t^*+\tau$ exploits all useful information in the future.

\subsection{Results}
Let us start with the case that $L_1=32$. Figure \ref{fig: DA_L32_final} shows the snapshot of the true flow field at $t=5$ (Panel (a)) and the recovered one based on the posterior mean estimate (Panel (b)) using these existing $L_1$ drifters. The figure also includes a comparison between the true signal of four different Fourier modes (blue) and the posterior mean time series (red), as well as the associated uncertainty in the posterior estimate (red shading). The uncertainty shown here is the two standard deviations of the posterior distribution at each time instant. With such a large number of existing drifters, the posterior mean captures the truth with relatively high accuracy. The posterior standard deviation is noticeable though it is not significant.

\begin{figure}[htb]\centering
  \hspace*{-2cm}\includegraphics[width=17cm]{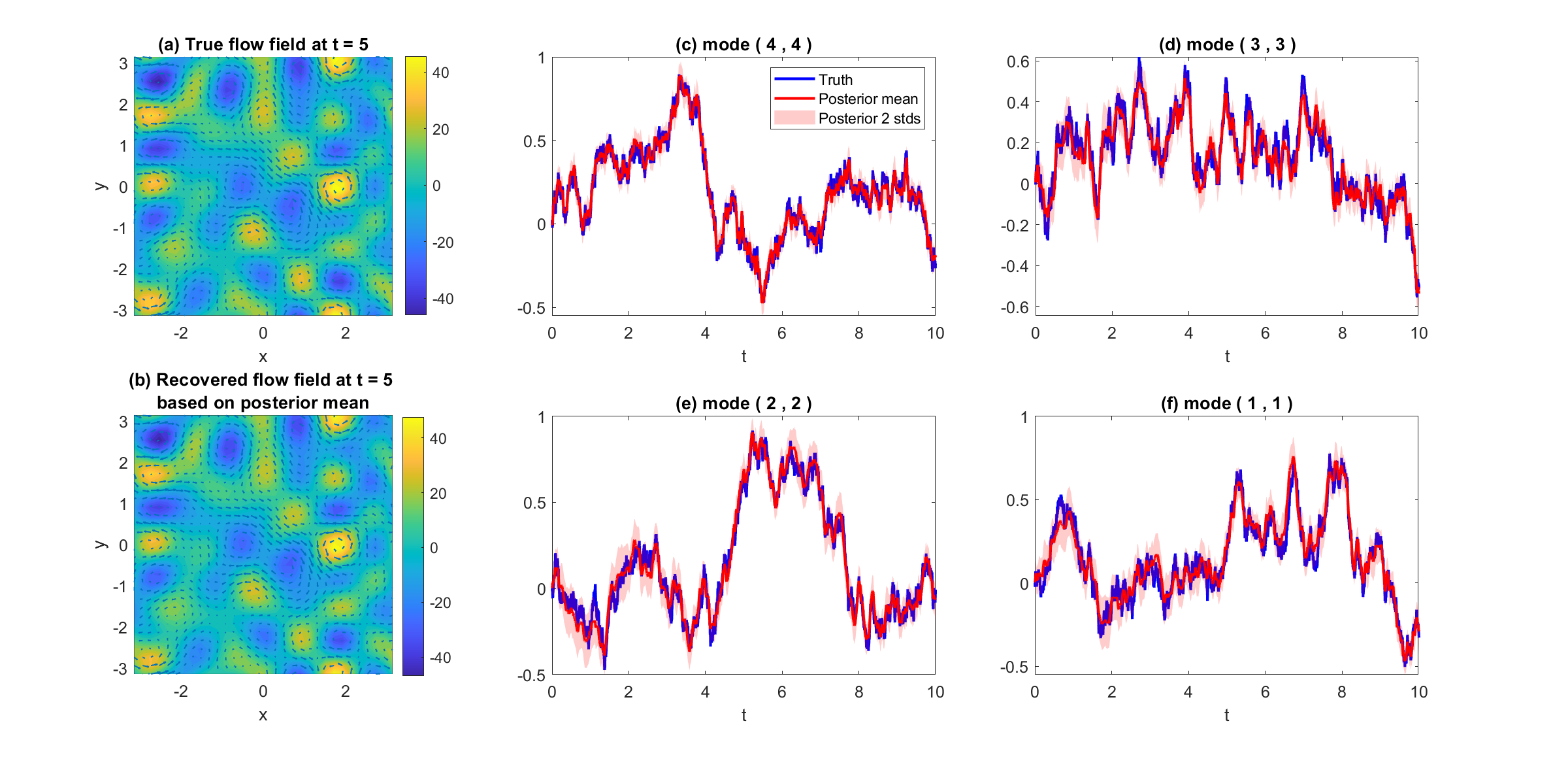}
  \caption{Comparison of the true flow field with the estimated one using the original $L_1$ drifters. Panel (a): The snapshot of the true flow field at $t=5$. Panel (b): The recovered flow field based on the posterior mean estimate. Panels (c)--(f): The true signal of four different Fourier modes (blue), the associated posterior mean time series (red), as well as the two standard deviations of the posterior distribution (red shading). In this figure, $L_1=32$ is used. }\label{fig: DA_L32_final}
\end{figure}

Figure \ref{fig: Tracer_locations_L32_final} aims to illustrate the importance of the distance criterion. Panel (a) shows the total information gain after the $L_2=6$ new drifters are launched into the flow field at the maxima resulting from the Lagrangian descriptor. The information gain is displayed as a function of the minimum distance $D_{min}$. Because of the relatively large number of drifters, there will be no feasible solutions for placing the $L_2$ new drifters when $D_{min}\geq 1$. Noticeably, the information gain remains high once $D_{min}\geq0.6$. Panels (b)--(h) display the locations of the $L_2$ new drifters launched at the areas corresponding to the maxima of the Lagrangian descriptor (shown in the background contour map) with different choices of $D_{min}$. The information gain increases overall as the distance $D_{min}$, which implies the necessity of considering the distance criterion. If $D_{min}$ is tiny, for example, in the case in Panel (b), then all the $L_2$ new drifters will be concentrated in a local area. These drifters will carry similar information in such a situation, and the total information gain is limited. The drifters become more separated as $D_{min}$ increases. Therefore, they can collect information on the global flow field.

\begin{figure}[htb]\centering
  \hspace*{-2cm}\includegraphics[width=17cm]{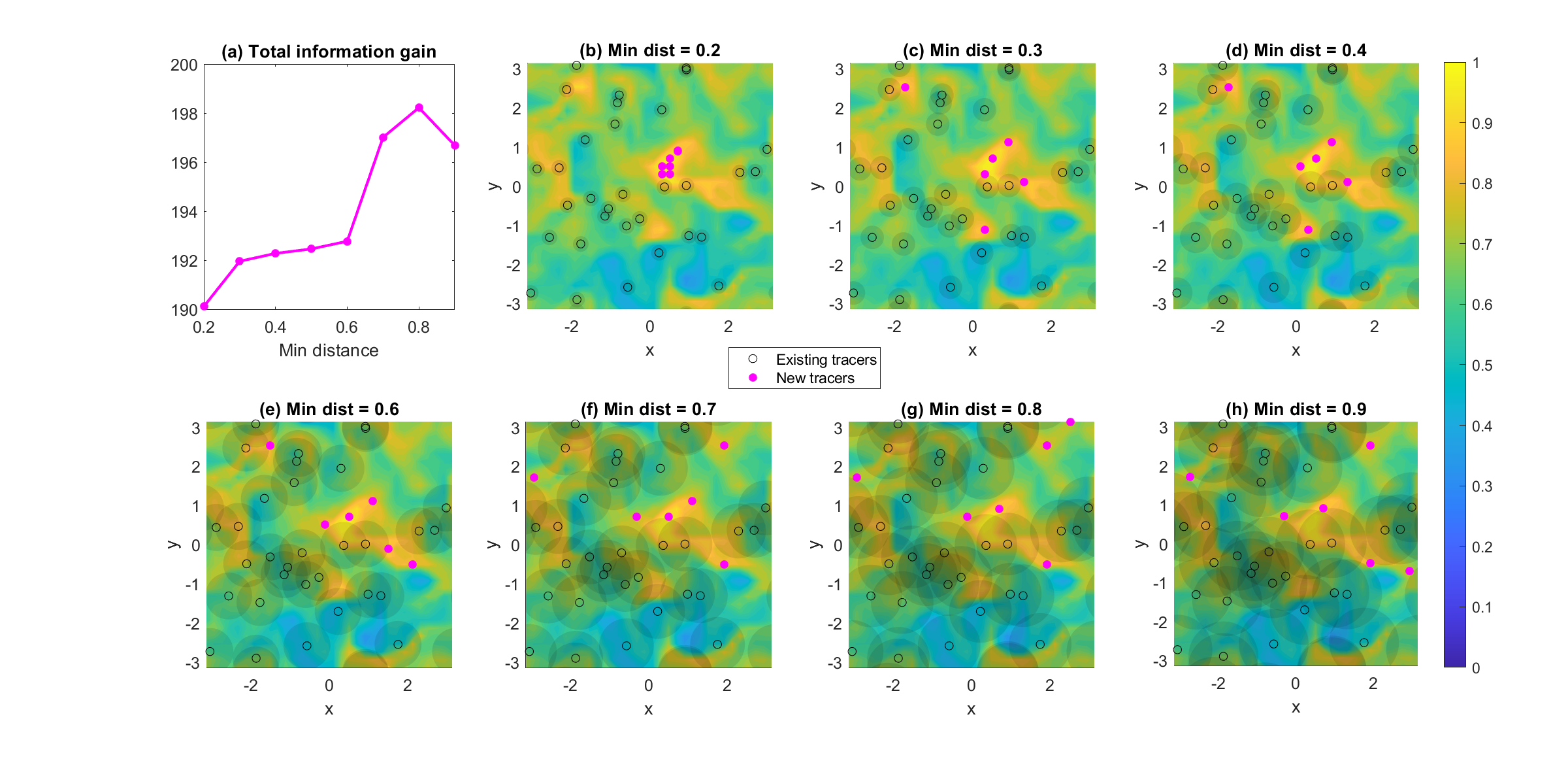}
  \caption{The total information gain using $L_1+L_2=32+6$ drifters, the locations of the existing $L_1$ drifters, and the newly added $L_2$ drifters with different minimum distance $D_{min}$. Panel (a): The total information gain using $L_1+L_2$ drifters as a function of $D_{min}$. Panels (b)--(h): The locations of the existing $L_1$ drifters (black dots) and the newly added $L_2$ ones (magenta dots) at $t=t^*=5$. The semi-transparent black shading circle around each existing drifter shows the area within the minimum distance $D_{min}$ of that drifter. The newly added $L_1$ drifters are also separated by at least $D_{min}$ from each other, but the circles are not shown here. The background contour in each panel is the two-dimensional map computed from the Lagrangian descriptor \eqref{LD_VelocityBased_UQ}, which is used as part of the criterion to determine the locations of the $L_2$ new drifters. }\label{fig: Tracer_locations_L32_final}
\end{figure}

Next, Figure \ref{fig: Compare_with_random_L32} compares the information gain based on the proposed strategy (the magenta line) and three other methods. It aims to illustrate the importance of using the Lagrangian descriptor to determine the locations of deploying the new drifters. All the simulations in this figure have $D_{min}=0.8$, and therefore the drifters are separate enough from each other. The green dots correspond to the method that the $L_2$ drifters are placed randomly, satisfying a uniform distribution, to the locations that guarantee all the $L_2$ new drifters are distanced by at least $D_{min}$ between each other and are at least $D_{min}$ from the existing $L_1$ ones. In other words, such a strategy still satisfies the distance criterion but does not exploit the result using the Lagrangian descriptor. On the other hand, the gray dots correspond to a different method, in which the location of each of the $L_2$ new drifters is given by a pair of random numbers drawn from a uniform distribution in the entire $[-\pi,\pi]^2$ domain. This somewhat allows the placed drifters to separate but does not precisely satisfy the proposed distance criterion. There are, in total, $100$ green dots and $100$ gray dots, which are experiments with independent random number generators. For completeness, the blue line shows the information gain when the new drifters are placed at the minima of the Lagrangian descriptor map instead of the maxima, where the distance criterion is still applied. As the magenta line is above most of the green dots and the blue line, the significance of seeking the maxima of the Lagrangian descriptor is justified. Notably, the magenta line is always above the gray dots, indicating that both criteria are essential. The dispersion is the main contributor to the total information gain among the two components. This is intuitive, as the posterior mean is already close to the truth signal. Therefore, the major part of the uncertainty reduction lies in the covariance. It is worth highlighting that the standard path-wise measurements cannot assess the uncertainty associated with the covariance part, implying the insufficiency in using these path-wise measures to determine the drifter locations or quantify the uncertainty reduction.

\begin{figure}[htb]\centering
  \hspace*{-2cm}\includegraphics[width=17cm]{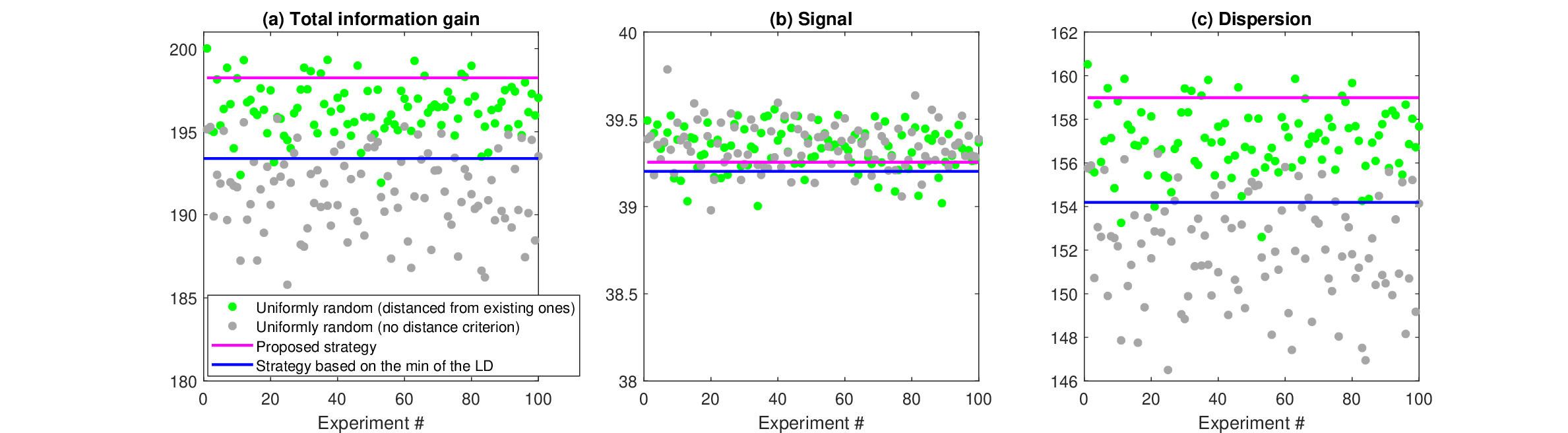}
  \caption{Comparison of the information gain based on the proposed strategy (the magenta line) and three other methods. Panel (a): The total information gain. Panels (b)--(c): The information gain in the signal and the dispersion parts. The magenta line shows the information gain using the proposed strategy. The green dots correspond to the method that the $L_2$ drifters are placed randomly, satisfying a uniform distribution, to the locations that guarantee all the $L_2$ new drifters are distanced by at least $D_{min}$ between each other and are from the existing $L_1$ drifters. In other words, such a strategy still satisfies the distance criterion but does not exploit the result using the Lagrangian descriptor. The gray dots correspond to a different method, in which the location of each of the $L_2$ new drifter is given by a pair of random numbers drawn from a uniform distribution in the entire $[-\pi,\pi]^2$ domain. This somewhat allows the placed drifters to separate but does not satisfy either of the proposed criteria. There are, in total, $100$ green dots and $100$ gray dots, which are experiments with independent random number generators. For completeness, the blue line shows the information gain when the new drifters are placed at the minima of the Lagrangian descriptor map instead of the maxima, where the distance criterion is still applied.  }\label{fig: Compare_with_random_L32}
\end{figure}

Figure \ref{fig: Example_random_L32_final} compares the drifter deploying strategies using different Lagrangian descriptors. Panel (a) shows the results when the deterministic version of the Lagrangian descriptor is used. Here, the deterministic Lagrangian descriptor is computed exploiting the true signal \eqref{LD_VelocityBased}, which is unknown in practice. In contrast, Panel (b) shows the one in the strategy developed here where the posterior uncertainty is incorporated \eqref{LD_VelocityBased_UQ}. One interesting finding is that, based on these two versions of the Lagrangian descriptor, the information gain (IG) in Panel (b) (IG$=198.224$) is higher than that in Panel (a) (IG$=196.549$). Since the maxima computed from the deterministic Lagrangian descriptor correspond to the drifters with the longest trajectories, the results highlight that letting drifters move the longest distances may not necessarily maximize the uncertainty reduction. In other words, computing the Lagrangian descriptor using the true flow field, which is unknown in practice, does not always provide the optimal result for uncertainty reduction. This is reasonable as the truth does not reflect the resulting uncertainty using the $L_1$ existing drifters. Such a finding is fundamentally different from those using strategies based on path-wise measures, which often aim to minimize errors with respect to the unknown truth. The Lagrangian descriptor incorporating the uncertainty guides the drifter deployment that reduces the uncertainty to a large extent. In both panels, the drifter deployment based on the green dot of Experiment \#1 in Figure \ref{fig: Compare_with_random_L32}, the best result from a random deployment satisfying the distance criterion, is also shown in green dots. These drifters fill in the areas with no existing ones more effectively. In addition, they are still mostly located at places with large values of the Lagrangian descriptor. Thus, this set of new drifters leads to slightly higher information gain.

\begin{figure}[htb]\centering
  \hspace*{-2cm}\includegraphics[width=17cm]{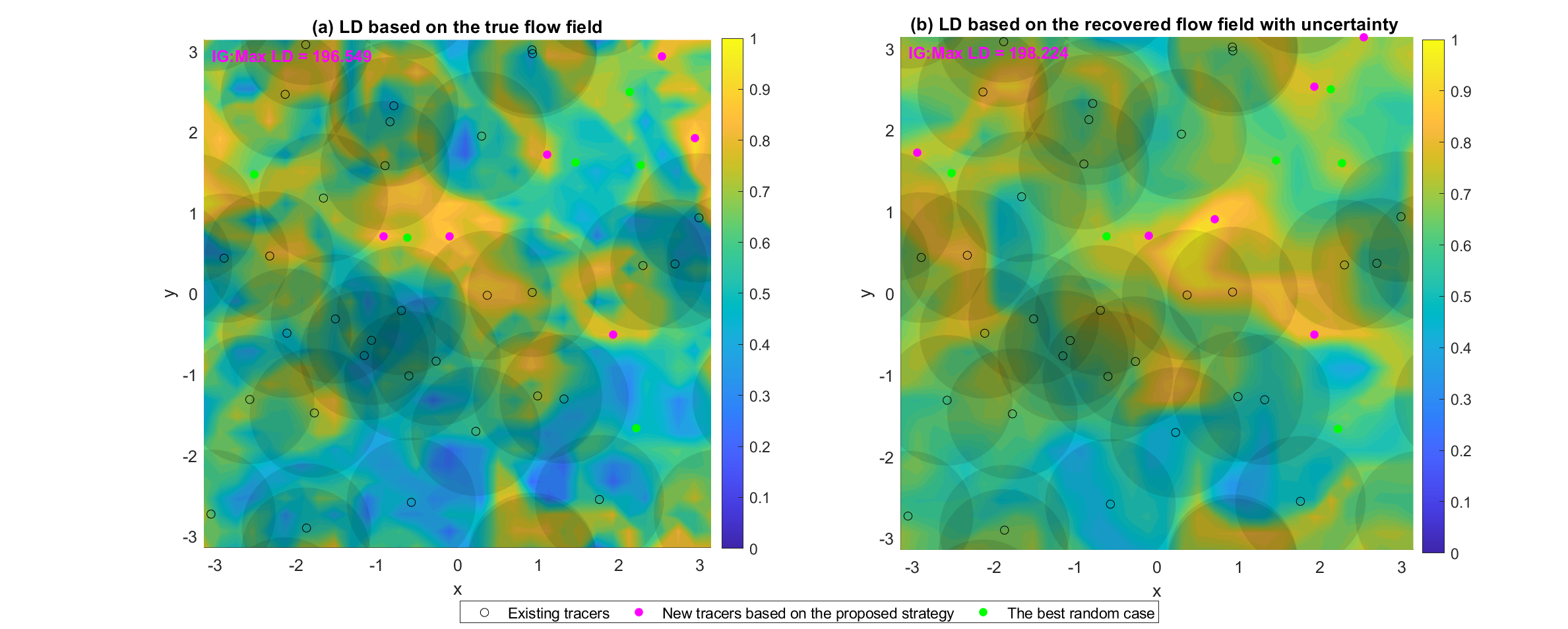}
  \caption{Comparison of the drifter deploying strategies using different Lagrangian descriptors. Panel (a): The results when the deterministic version of the Lagrangian descriptor is used. Here, the deterministic Lagrangian descriptor is computed exploiting the true signal \eqref{LD_VelocityBased}, which is unknown in practice. Panel (b): The one in the strategy developed here where the posterior uncertainty is incorporated \eqref{LD_VelocityBased_UQ}. In both panels, the locations of the new $L_2$ drifters using the proposed strategy based on the maxima of the Lagrangian descriptor (magenta dots) are shown. The drifter deployment based on the green dot of Experiment \#1 in Figure \ref{fig: Compare_with_random_L32}, the best result from a random deployment satisfying the distance criterion, is also shown in green dots. The information gain (IG) corresponding to the strategies of putting drifters at the maxima of the Lagrangian descriptor is listed on the left of each panel.    }\label{fig: Example_random_L32_final}
\end{figure}

Figures \ref{fig: DA_L12_final}--\ref{fig: Compare_with_random_L12} show the results when $L_1=12$ drifters are initially included. Due to the diminishing of the existing drifter numbers, the recovered flow field is less accurate, which leads to larger uncertainty. See Figure \ref{fig: DA_L12_final}. Figure \ref{fig: Tracer_locations_L12_final} shows the information gain as a function of $D_{min}$. Since the number of the existing drifters is decreased,  $D_{min}$ can be as large as $D_{min}=1.2$. Similar to the results in Figure \ref{fig: Tracer_locations_L32_final}, the information gain remains at a high level once $D_{min}$ is larger than a certain value (which is $D_{min}=0.7$ in this case), indicating the robustness of the strategy. Likewise, the results in Figure \ref{fig: Compare_with_random_L12} confirm that the proposed strategy outperforms the other methods, consistent with the conclusion in Figure \ref{fig: Compare_with_random_L32}. Finally, qualitatively similar conclusions are reached when the new drifters are deployed at other time instants $t^*$ (not shown here), which reinforces the robust performance of the proposed strategy.

\begin{figure}[htb]\centering
  \hspace*{-2cm}\includegraphics[width=17cm]{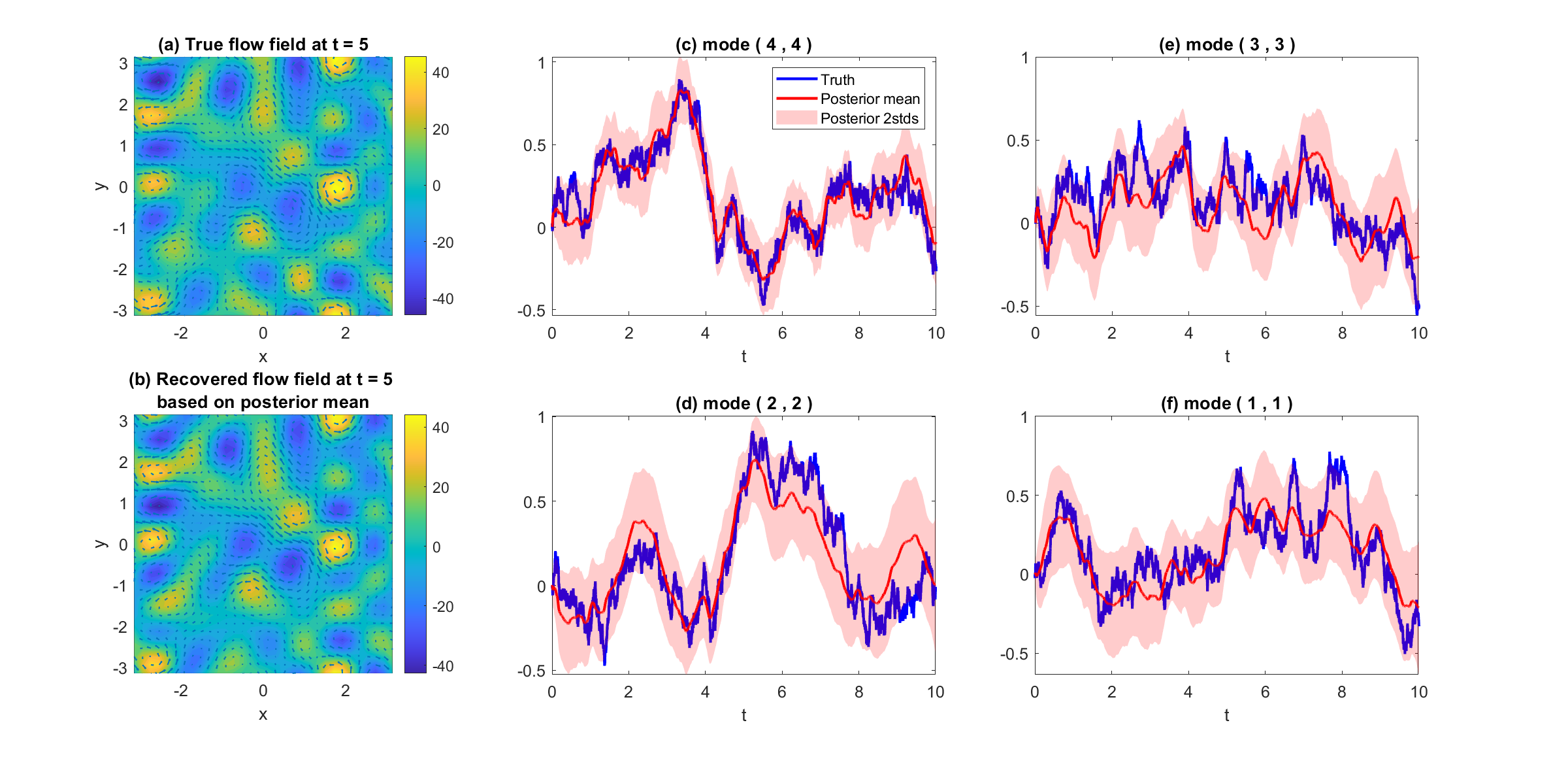}
  \caption{Similar to Figure \ref{fig: DA_L32_final} but with $L_1=12$. }\label{fig: DA_L12_final}
\end{figure}

\begin{figure}[htb]\centering
  \hspace*{-2cm}\includegraphics[width=17cm]{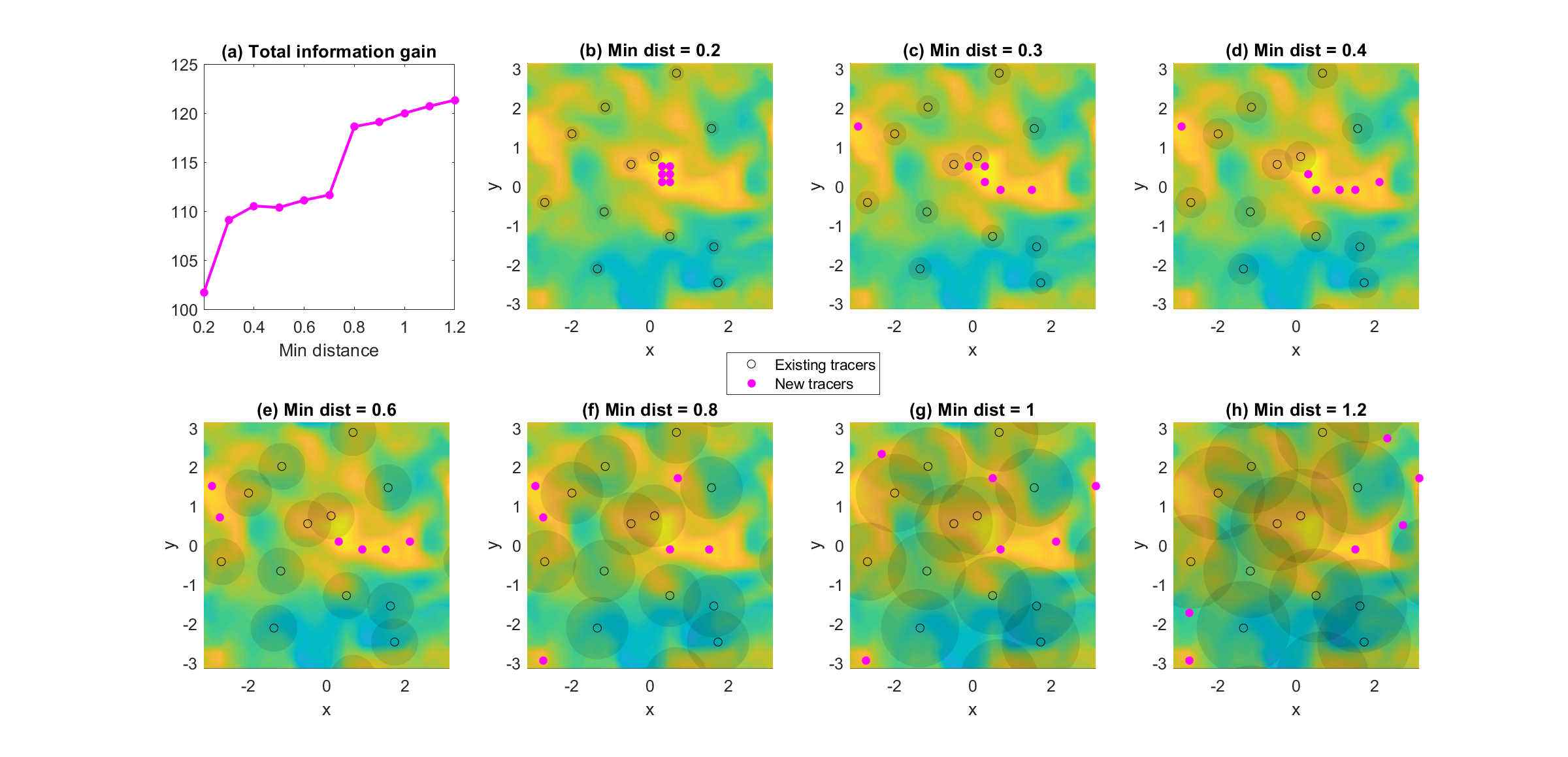}
  \caption{Similar to Figure \ref{fig: Tracer_locations_L12_final} but with $L_1=12$.}\label{fig: Tracer_locations_L12_final}
\end{figure}

\begin{figure}[htb]\centering
  \hspace*{-2cm}\includegraphics[width=17cm]{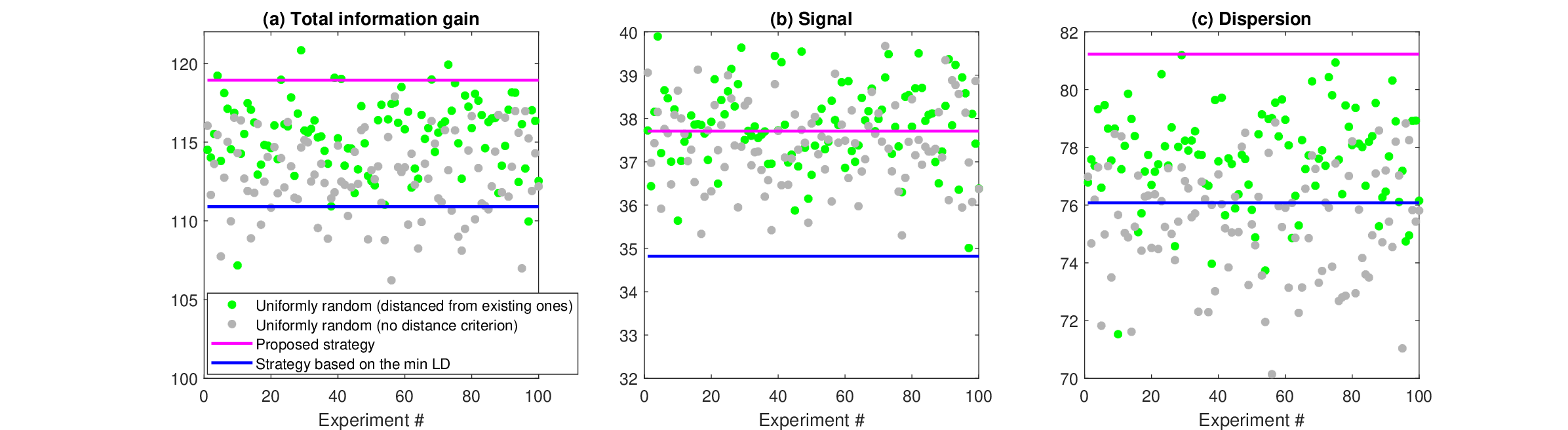}
  \caption{Similar to Figure \ref{fig: Compare_with_random_L32} ($D_{min}=0.8$) but with $L_1=12$.  }\label{fig: Compare_with_random_L12}
\end{figure}

\section{Discussions and Conclusion}\label{Sec:Conclusion}
In this paper, a new computationally efficient strategy for deploying Lagrangian drifters that highlights the central role of uncertainty is developed. A nonlinear trajectory diagnostic approach that underlines the importance of uncertainty is utilized to construct a phase portrait map. It consists of both the geometric structure of the underlying flow field and the uncertainty in the recovered state. The drifters are deployed at the maxima of this map and are required to be separated enough. Such a strategy allows the drifters to travel the longest distances to collect both the local and global information of the flow field. It also facilitates the reduction of a significant amount of uncertainty. An information metric is introduced to assess the performance of the strategy. Fundamentally different from the traditional path-wise measurements, the information metric quantifies the information captured by the entire estimated distribution that naturally considers the uncertainty reduction. The information metric also avoids using the unknown truth for assessing the skill of the strategy, making the method practical. Mathematical analysis exploiting simple illustrative examples is used to validate the strategy developed here. Numerical simulations based on multiscale turbulent flows are then adopted to demonstrate the practical performance of the new strategy.

The main focus of this work is to provide a systematic understanding of the strategy. Therefore, the same interval $[t^*-\tau, t^*+\tau]$ is utilized to compute the Lagrangian descriptor and assess the information gain after the new drifters are deployed at time $t^*$. Although the setup used in the simulations of this work is not real-time, the strategy can be easily extended to real-time forecast situations. Assume the existing $L_1$ drifter trajectories are available up to the current time instant $t^*$. The filtering solution \eqref{eq:filter} at $t^*$ is used as the initial condition to forecast the ocean field up to a future time instant $t^*+\tau$ using an ensemble forecast method. Then the Lagrangian descriptor \eqref{LD_VelocityBased_UQ} considering the forecast uncertainty is computed within the future interval $[t^*,t^*+\tau]$, which indicates the distance traveled by each potential trajectory starting from different locations at $t^*$. The maxima of the resulting map are the locations to deploy the $L_2$ new drifters at the current time $t^*$ that maximize the information collected by these drifters within the finite future time interval $[t^*,t^*+\tau]$.

Recall in this work that all the $L_2$ new drifters are deployed at the same time. Another future direction is to deploy the drifters sequentially. This can be done by deploying one drifter each time and then recomputing the Lagrangian descriptor to rectify the resulting map for launching the next drifter. This will guarantee that the deploying strategy for each drifter is optimal. The sequential strategy requires repeatedly computing the Lagrangian descriptor and the data assimilation (or forecast) solutions. Nevertheless, if a computationally efficient data assimilation, such as the one used in this work, is available, then the computational cost will not increase significantly. It is interesting to study the additional benefit of using such a sequential method for deploying the drifters. Finally, only incompressible flows are used in this study. Understanding the skill of the strategy developed here for compressible flow fields will be another interesting topic for future research.

\section*{Acknowledgement}
The research of N.~C.~ is funded by ONR N00014-19-1-2421 and ARO W911NF-23-1-0118. The research of E.~L.~ is supported by ONR N0001423WX01622. S~.W.~acknowledges the financial support provided by the EPSRC Grant No. EP/P021123/1 and the support of the William R. Davis '68 Chair in the Department of Mathematics at the United States Naval Academy.



\bibliographystyle{elsarticle-num}





\end{document}